\documentclass{article}
\usepackage[margin=1in]{geometry}

\usepackage[utf8]{inputenc}
\usepackage[T1]{fontenc}
\usepackage{lmodern}

\usepackage{amsmath,amssymb, amsfonts, amsthm, mathtools, bm} 
\usepackage{empheq}
\newtheorem{theorem}{Theorem}

\newtheorem{corollary}[theorem]{Corollary}
\newtheorem{proposition}[theorem]{Proposition}
\newtheorem{lemma}[theorem]{Lemma}

\usepackage{tabularx}
\usepackage{enumerate}
\usepackage{graphicx}
\usepackage{euscript}
\usepackage{xspace}

\usepackage{calc}
\usepackage{tikz}
\usetikzlibrary{calc}
\usetikzlibrary{decorations.markings}
\usetikzlibrary{arrows,calc,decorations.pathmorphing}
\usetikzlibrary{decorations.pathreplacing,shapes.misc}
\usetikzlibrary{positioning}
\usetikzlibrary{arrows.meta}

\usepackage{xcolor}
\definecolor{orange}{rgb}{0.8,0.4,0}
\definecolor{violet}{rgb}{0.6,0,0.8}
\definecolor{darkgreen}{rgb}{0,0.5,0}
\definecolor{verydarkgreen}{rgb}{0.0,0.3,0.0}
\definecolor{darkblue}{rgb}{0,0,0.6}
\definecolor{darkred}{rgb}{0.75,0,0}
\definecolor{grey}{rgb}{0.35,0.35,0.35}

\DeclareMathOperator{\dist}{dist}

\renewcommand{\epsilon}{\varepsilon}

\newcommand{\BRD}{BonnRoute\xspace}

\def\Rbb{\mathbb{R}}

\def\E{\textsc{e}}
\def\N{\textsc{n}}
\def\W{\textsc{w}}
\def\S{\textsc{s}}
\def\NE{\textsc{ne}}
\def\NW{\textsc{nw}}
\def\SW{\textsc{sw}}
\def\SE{\textsc{se}}

\usepackage[%
 pdfpagelabels,
 hyperindex,
 bookmarksopen=true,
 bookmarksopenlevel=1,
]{hyperref}

\hypersetup{
    colorlinks=true,
    linkcolor=darkblue,
    citecolor = darkgreen,
}

\newcommand\blfootnote[1]{%
  \begingroup
  \renewcommand\thefootnote{}\footnote{#1}%
  \addtocounter{footnote}{-1}%
  \endgroup
}

\makeatletter
\let\@fnsymbol\@arabic
\makeatother

\author{Markus Ahrens\thanks{
IBM Deutschland Research \& Development GmbH.
\texttt{markus.johannes.ahrens@ibm.com}
}
\and 
Dorothee Henke\thanks{
Department of Mathematics, TU Dortmund University.
\texttt{dorothee.henke@math.tu-dortmund.de}
}
\and 
Stefan Rabenstein\thanks{
Research Institute for Discrete Math., Hausdorff Center for Mathematics, University of Bonn.
\texttt{rabenstein@dm.uni-bonn.de}
}
\and 
Jens Vygen\thanks{
Research Institute for Discrete Math., Hausdorff Center for Mathematics, University of Bonn. 
\texttt{vygen@dm.uni-bonn.de}
}
}

\date{}

\title{Faster Goal-Oriented Shortest Path Search \\ for Bulk and Incremental Detailed Routing}

\begin{document}

\maketitle

\blfootnote{
Corresponding Author: Stefan Rabenstein \\
Affiliation: Research Institute for Discrete Math., Hausdorff Center for Mathematics, University of Bonn. \\
E-mail Address: \texttt{rabenstein@dm.uni-bonn.de}
}

\begin{abstract}
We develop new algorithmic techniques for VLSI detailed routing.
First, we improve the goal-oriented version of Dijkstra's algorithm to find shortest paths
in huge incomplete grid graphs with edge costs depending on the direction and the layer, and possibly on rectangular regions. 
We devise estimates of the distance to the targets that offer better trade-offs between
running time and quality than previously known methods, leading to an overall speed-up.
Second, we combine the advantages of the two classical detailed routing approaches ---
global shortest path search and track assignment with local corrections --- 
by treating input wires (such as the output of track assignment) as reservations that can be used at a discount by the
respective net. We show how to implement this new approach efficiently. 
\end{abstract}

\section{Introduction}
\begin{figure}[t]
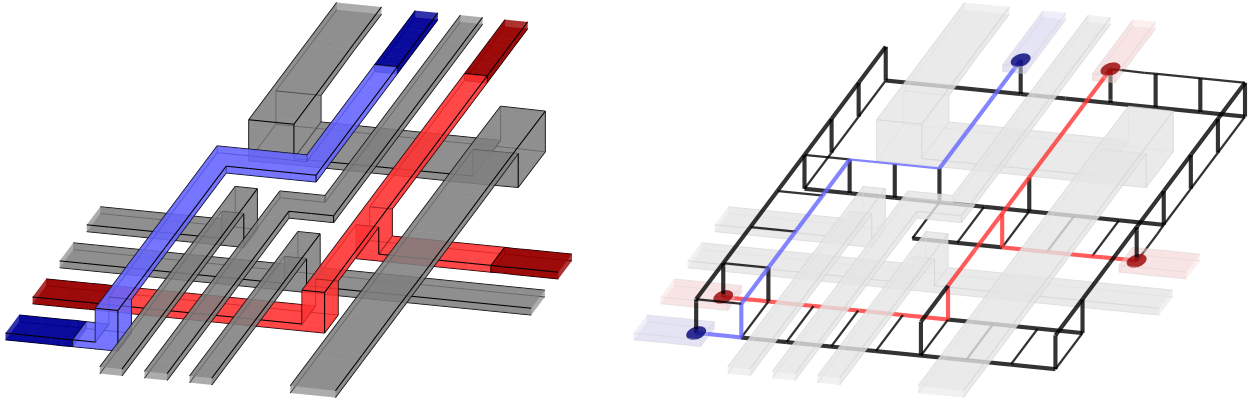

 \centering
 \begin{tikzpicture}[scale=0.6]
  \input{example}
 \end{tikzpicture}
 \hfill
\begin{tikzpicture}[scale=0.6]
  \input{example_graph}
 \end{tikzpicture}
 \caption{\label{fig:routedchipexample}Left: tiny part of a routed chip with two layers. The blue wires connect the two dark blue pins and the red wires connect the three dark red pins. Pins of the same color belong to the same net. The gray wires are part of the connections of multiple other nets with pins outside of the visible region.
 Right: the relevant part of the detailed routing graph before routing the blue and the red net. Steiner trees for the red net and the blue net.}
\end{figure}

The task of VLSI routing \cite{AlpertMehtaSapatnekar2008, Tellez2016}
is to connect the set of pins of every net on a chip by wires so that wires of different nets
are sufficiently far apart and various other constraints are met.
See Figure \ref{fig:routedchipexample} (left) for an example.
In a simple but useful model, we have a huge 3-dimensional grid graph 
(the \emph{detailed routing graph}), and the \emph{pins} are vertices in this graph.
Each \emph{net} is a set of pins and needs to be connected by a Steiner tree in the detailed routing graph.
The Steiner trees of distinct nets must be vertex-disjoint.
The detailed routing graph is induced by \emph{routing tracks}, which are pre-computed parallel lines on each layer. 
Every two routing tracks on adjacent layers are orthogonal to each other and induce one vertex on each of these tracks.
These two vertices are connected by an edge; 
metal connections along those edges (connecting adjacent layers) are called \emph{vias}. 
Depending on the manufacturing process, vertices on adjacent tracks of the same layer may also be connected by an edge.  

Typically, one first computes a \emph{global routing}, a rough packing of wires that ignores all local constraints
but guarantees that the wires in certain areas do not require more space than available.
Global routing allows for globally optimizing objectives such as power consumption and timing constraints
\cite{MuellerRadkeVygen2011,Heldetal2018}.

The output of global routing then restricts the search space for every net in detailed routing, where
many complicated rules need to be obeyed and one essentially routes one net at a time.
While the detailed routing graph on an entire chip can contain about
$10^{13}$ vertices on 10--20 layers, the restricted area corresponding to the global routing solution for a net
results in a much smaller detailed routing graph, with rarely more than $10^{8}$ vertices. 
Nevertheless, these subgraphs are still huge, and there are millions of nets to connect.
Two general strategies have been proposed (cf.\ \cite{Tellez2016}). 

The first approach is based on a fast subroutine to find a shortest path that connects two metal components, each of which
can consist of a pin or a set of previously computed wires connecting a subset of the pins of that net. 
The subgraph is given by the global routing solution, excluding vertices and edges that would result in a conflict to previously routed wires.
For an example of the resulting graph, see Figure \ref{fig:routedchipexample} (right).
To allow for an efficient packing of wires and to model various aspects such as signal delays, one uses  
different costs for horizontal and for vertical edges on each layer as well as for vias.

The second approach first considers the layers one after the other and assigns wires to routing tracks so that
the most important detailed routing rules are satisfied, at least for most wires. 
This is often called \emph{track assignment} \cite{sarrafzadeh1994restricted,batterywala2002track}.
Then detailed routing tries to correct violations locally. 
A very similar detailed routing problem occurs when a detailed routing has already been computed, but a few changes
to the input have been made (for example corrections of the logical behavior or to speed up signals that arrived too late).
In both cases, one asks for an \emph{incremental} detailed routing, largely following the input but deviating where necessary.
However, local corrections are often not possible if the routing is very dense.

Goal-oriented path search (sometimes called $A^*$) is a
classical speed-up technique of Dijkstra's shortest path algorithm \cite{dijkstra}. It is based on a feasible potential that estimates the distance to the targets \cite{HarNR68,Rub74,LawLP83}. 
Instead of the undirected graph with the original edge cost $c(e)$,
we orient each edge in both ways and run Dijkstra's algorithm with the \emph{reduced cost} 
$c_{\pi}(e):=c(e)-\pi(v)+\pi(w)$ for every edge $e$ directed from $v$ to $w$,
where the vertex potentials~$\pi$ are chosen so that $c_{\pi}$ is nonnegative and $\pi(t)=0$ for every target $t$.
These conditions imply that $\pi(v)$ is a lower bound on the distance between $v$ and the closest target.
The better this lower bound is, the fewer vertices this goal-oriented version of Dijkstra's algorithm must label before it knows a shortest path
to a target.

Hence, there is a trade-off between a possible preprocessing time, the time to compute
the potential of a vertex, and the quality of the lower bound.
For example, in subgraphs of unweighted grid graphs, the $\ell_1$-distance to the nearest target can be a reasonable choice for $\pi$ \cite{Hetzel1998}. 
A better estimate, which however requires substantial preprocessing, was suggested by \cite{PeyerRautenbachVygen2009}. 
In this paper, we propose new methods with better trade-offs than previously known.

Moreover, we combine the advantages of the two classical detailed routing approaches mentioned above.
Our new, more global approach treats given input wires (e.g., the output of track assignment) 
as so-called \emph{reservations}. A reservation for a net $N$ is a set of edges reserved for $N$ until $N$ is routed: no other net must use these edges.
We encourage, but not force, the detailed router to follow the reservations where feasible.
This is achieved by finding a shortest path where reservations of the currently routed net can be used at a discount 
(so we reduce the cost of reserved edges by a fixed factor smaller than 1).

However, this does not work well together with the classical goal-oriented techniques. For example,
if there are some reservations (edges) that can be used at a 50\,\% discount, 
the $\ell_1$-distance would have to be divided by 2 in order to induce a feasible potential.
This would often be a very inaccurate estimate, leading to an increased number of labels in Dijkstra's algorithm and hence larger running time.
We show that our better potentials make goal-oriented Dijkstra not only as fast as without reservations,
but in fact faster. Overall, this yields a new efficient incremental detailed routing algorithm.

\subsection{Problem statement}

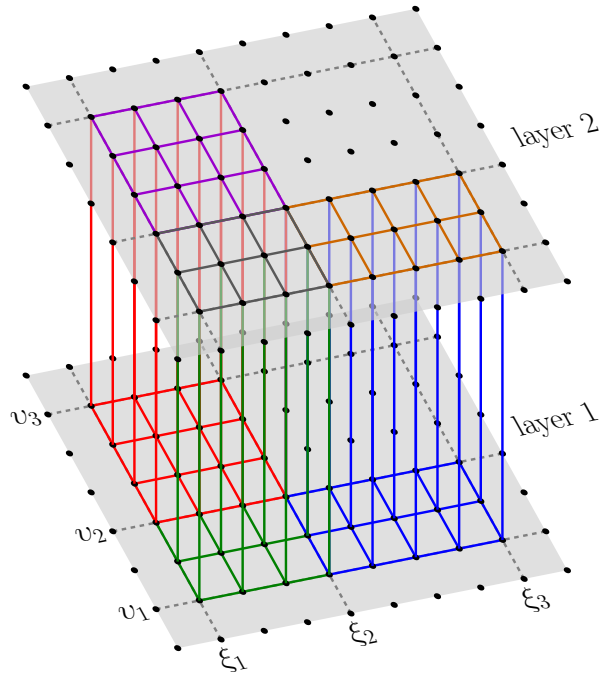
\begin{figure}[b]
 \centering
   \resizebox{0.5\textwidth}{!}
  {
  \begin{tikzpicture}
   \newcommand{\yslant}{0.2} %0.5
   \newcommand{\xslant}{-0.5} %-1
   \newcommand{\yshift}{190} %379
   \newcommand{\opa}{0.6}

   % Variables for coordinates
   \newcommand{\xfirst}{0}
   \newcommand{\xione}{1}
   \newcommand{\xitwominus}{3}
   \newcommand{\xitwo}{4}
   \newcommand{\xithreeminus}{7}
   \newcommand{\xithree}{8}
   \newcommand{\xlast}{9}

   \newcommand{\yfirst}{0}
   \newcommand{\upsone}{1}
   \newcommand{\upstwominus}{2}
   \newcommand{\upstwo}{3}
   \newcommand{\upsthreeminus}{5}
   \newcommand{\upsthree}{6}
   \newcommand{\ylast}{7}

   % Layer 1
  \begin{scope}[
 yshift=-\yshift,every node/.append style={
 yslant=\yslant,xslant=\xslant},yslant=\yslant,xslant=\xslant
  ]
   \draw ($(\xlast + 1.5, \ylast / 2 + \yfirst / 2)$) node {\fontsize{20}{25}\selectfont layer $1$};
   \fill[black!20,fill opacity=\opa] (\xfirst,\yfirst) rectangle (\xlast,\ylast);

   % Hanan grid
   \foreach \x in { \xione, \xitwo, \xithree }
   {
     \draw[ultra thick, gray, dashed] (\x,\yfirst) -- (\x,\ylast);
   }
   \foreach \y in { \upsone, \upstwo, \upsthree }
   {
     \draw [ultra thick, gray, dashed] (\xfirst,\y) -- (\xlast,\y);
   }

   \draw (\xione,\yfirst-0.5) node {\fontsize{20}{25}\selectfont $\xi_1$};
   \draw (\xitwo,\yfirst-0.5) node {\fontsize{20}{25}\selectfont $\xi_2$};
   \draw (\xithree,\yfirst-0.5) node {\fontsize{20}{25}\selectfont $\xi_3$};

   \draw (\xfirst-0.5,\upsone) node {\fontsize{20}{25}\selectfont $\upsilon_1$};
   \draw (\xfirst-0.5,\upstwo) node {\fontsize{20}{25}\selectfont $\upsilon_2$};
   \draw (\xfirst-0.5,\upsthree) node {\fontsize{20}{25}\selectfont $\upsilon_3$};

   % horizontal red edges
   \foreach \x in {\xione, ..., \xitwominus}
   {
     \foreach \y in {\upstwo, ..., \upsthree}
     {
       \draw[red,ultra thick] (\x,\y) -- (\x+1,\y);
     }
   }
   % vertical red edges
   \foreach \x in {\xione, ..., \xitwo}
   {
     \foreach \y in {\upstwo, ..., \upsthreeminus}
     {
       \draw[red,ultra thick] (\x,\y) -- (\x,\y+1);
     }
   }
   % horizontal green edges
   \foreach \x in {\xione, ..., \xitwominus}
   {
     \foreach \y in {\upsone, ..., \upstwominus}
     {
       \draw[darkgreen,ultra thick] (\x,\y) -- (\x+1,\y);
     }
   }
   % vertical green edges
   \foreach \x in {\xione, ..., \xitwominus}
   {
     \foreach \y in {\upsone, ..., \upstwominus}
     {
       \draw[darkgreen,ultra thick] (\x,\y) -- (\x,\y+1);
     }
   }
   % horizontal blue edges
   \foreach \x in {\xitwo, ..., \xithreeminus}
   {
     \foreach \y in {\upsone, ..., \upstwo}
     {
       \draw[blue,ultra thick] (\x,\y) -- (\x+1,\y);
     }
   }
   % vertical blue edges
   \foreach \x in {\xitwo, ..., \xithree}
   {
     \foreach \y in {\upsone, ..., \upstwominus}
     {
       \draw[blue,ultra thick] (\x,\y) -- (\x,\y+1);
     }
   }
   % vertices
   \foreach \x in {\xfirst, ..., \xlast}
   {
     \foreach \y in {\yfirst, ..., \ylast}
     {
       \node[fill, circle, inner sep=1.5pt] at (\x,\y) {};
     }
   }
   
  \end{scope}
   
   %    Vias
  \begin{scope}[
 yshift=-\yshift,every node/.append style={
 yslant=\yslant,xslant=\xslant},yslant=\yslant,xslant=\xslant
  ]
  % blue vias
      \foreach \x in {\xitwo, ..., \xithree}
   {
     \foreach \y in {\upsone, ..., \upstwo}
     {
       \draw [blue, ultra thick] (\x,\y) -- ($(\x, \y) + 0.5*(6.675,13.35)$);
     }
   }
  % red vias
  \foreach \x in {\xione, ..., \xitwo}
  {
    \foreach \y in {\upstwo, ..., \upsthree}
    {
      \draw [red, ultra thick] (\x,\y) -- ($(\x, \y) + 0.5*(6.675,13.35)$);
    }
  }
   % green vias
   \foreach \x in {\xione, ..., \xitwo}
   {
     \foreach \y in {\upsone, ..., \upstwominus}
     {
       \draw [darkgreen, ultra thick] (\x,\y) -- ($(\x, \y) + 0.5*(6.675,13.35)$);
     }
   }
  \end{scope}
   
  % Layer 2
  \begin{scope}[
 yshift=0,every node/.append style={
 yslant=\yslant,xslant=\xslant},yslant=\yslant,xslant=\xslant
  ]
  \draw ($(\xlast + 1.5, \ylast / 2 + \yfirst / 2)$) node {\fontsize{20}{25}\selectfont layer $2$};
   \fill[black!20,fill opacity=\opa] (\xfirst,\yfirst) rectangle (\xlast,\ylast);

   % Hanan grid
   \foreach \x in { \xione, \xitwo, \xithree }
   {
     \draw[ultra thick, gray, dashed] (\x,\yfirst) -- (\x,\ylast);
   }
   \foreach \y in { \upsone, \upstwo, \upsthree }
   {
     \draw [ultra thick, gray, dashed] (\xfirst,\y) -- (\xlast,\y);
   }

   % horizontal violet edges
   \foreach \x in {\xione, ..., \xitwominus}
   {
     \foreach \y in {\upstwo, ..., \upsthree}
     {
       \draw[violet,ultra thick] (\x,\y) -- (\x+1,\y);
     }
   }
   % vertical violet edges
   \foreach \x in {\xione, ..., \xitwo}
   {
     \foreach \y in {\upstwo, ..., \upsthreeminus}
     {
       \draw[violet,ultra thick] (\x,\y) -- (\x,\y+1);
     }
   }
   % horizontal orange edges
   \foreach \x in {\xitwo, ..., \xithreeminus}
   {
     \foreach \y in {\upsone, ..., \upstwo}
     {
       \draw[orange,ultra thick] (\x,\y) -- (\x+1,\y);
     }
   }
   % vertical orange edges
   \foreach \x in {\xitwo, ..., \xithree}
   {
     \foreach \y in {\upsone, ..., \upstwominus}
     {
       \draw[orange,ultra thick] (\x,\y) -- (\x,\y+1);
     }
   }
   % horizontal grey edges
   \foreach \x in {\xione, ..., \xitwominus}
   {
     \foreach \y in {\upsone, ..., \upstwo}
     {
       \draw[grey,ultra thick] (\x,\y) -- (\x+1,\y);
     }
   }
   % vertical grey edges
   \foreach \x in {\xione, ..., \xitwo}
   {
     \foreach \y in {\upsone, ..., \upstwominus}
     {
       \draw[grey,ultra thick] (\x,\y) -- (\x,\y+1);
     }
   }

   % vertices
   \foreach \x in {\xfirst, ..., \xlast}
   {
     \foreach \y in {\yfirst, ..., \ylast}
     {
       \node[fill, circle, inner sep=1.5pt] at (\x,\y) {};
     }
   }
   
  \end{scope}
  \end{tikzpicture}
  }
  \caption 
  {
    \label{fig:generalcostmodel}
    Example of the general cost model with $p = 3$, $q = 3$, and $l = 2$.
    The \textcolor{red}{red} edges form the set~$\textcolor{red}{E_1^{12}}$.
    The \textcolor{blue}{blue} horizontal edges have cost $\textcolor{blue}{c_1^{21\leftrightarrow}}$,
    the \textcolor{blue}{blue} vertical edges have cost $\textcolor{blue}{c_1^{21\updownarrow}}$,
    and the \textcolor{blue}{blue} via edges have cost $\textcolor{blue}{c_{1,2}^{21}}$.
    The cost of the \textcolor{blue}{blue} vertical edges on $\xi_2$ is given by
    $\min\{\textcolor{darkgreen}{c_1^{11\updownarrow}}, \textcolor{blue}{c_1^{21\updownarrow}}\}$,
    which is $\textcolor{blue}{c_1^{21\updownarrow}}$ in this example.
    The cost of the \textcolor{darkgreen}{green} via edges on $\xi_2$ is
    $\min\{\textcolor{darkgreen}{c_{1,2}^{11}}, \textcolor{blue}{c_{1,2}^{21}}\}$,
    which is $\textcolor{darkgreen}{c_{1,2}^{11}}$ here.
    The edges that are not drawn have cost infinity because, in this example, they lie outside the area corresponding to the global routing solution.
  }
\end{figure}

Our core problem will consist of computing distances in a weighted grid graph with a simple structure.
To define the grid graph,
we number the layers $1,\ldots,l$ and let $V=\mathbb{Z}\times\mathbb{Z}\times\{1,\ldots,l\}$
and $$E=\left\{\{(x,y,z),(x',y',z')\}\in{\textstyle\binom{V}{2}} \mid |x-x'|+|y-y'|+|z-z'|=1\right\}$$ be the vertex set and edge set
of an infinite grid with $l$ layers. Edges connecting adjacent layers are called \emph{vias}, edges in x-direction are \emph{horizontal} and edges in y-direction \emph{vertical}.
We will consider finite subgraphs of $G=(V,E)$.
These subgraphs correspond to the area defined by the global routing solution and
to the restriction to the routing tracks that can be used for the current net. 
Often, many vertices of these subgraphs will have degree 2 and will not be considered explicitly,
but we ignore this here for the sake of a simpler exposition. 

Every layer has a preference direction ($\leftrightarrow$ or $\updownarrow$, the direction of its tracks); 
edges in the other direction are more expensive or sometimes even forbidden, depending on the manufacturing process.
Horizontal and vertical layers alternate.
Moreover, the layers have very different electrical properties, which is reflected by appropriate edge costs. 
In the simplest model, the cost of an edge depends only on its direction and the layer: 
let $c_{z}^{\leftrightarrow},c_{z}^{\updownarrow}>0$ for $z\in\{1,\ldots,l\}$ and $c_{z,z+1}>0$ for $z\in\{1,\ldots,l-1\}$; then
$$c(\{(x,y,z),(x',y',z')\}) = \begin{cases} 
c_{z}^{\leftrightarrow} & \text{if } x'=x+1 \\[1mm] 
c_{z}^{\updownarrow} & \text{if } y'=y+1 \\ 
c_{z,z'} & \text{if } z'=z+1 
\end{cases}.$$

In a more general model, a rectilinear grid induces rectangular regions,
called \emph{tiles}, and the cost also depends on the tile.
Let $$-\infty=\xi^0<\xi^1\leq\ldots\leq\xi^p<\xi^{p+1}=\infty, \qquad-\infty=\upsilon^0<\upsilon^1\leq\ldots\leq\upsilon^q<\upsilon^{q+1}=\infty$$ 
be integer coordinates that define the rectangular tiles
$$V^{ij}_z=\left\{(x,y,z)\in V \mid \xi^{i}\le x\le\xi^{i+1},\, \upsilon^{j}\le y\le\upsilon^{j+1}\right\},$$
and set
$$E^{ij}_{z}=\left\{\{(x,y,z),(x',y',z')\}\in E \mid \xi^{i} \le x \le x' \le \xi^{i+1},\, \upsilon^{j} \le y \le y' \le \upsilon^{j+1},\, z\le z'\right\}.$$
Now we have costs $c_{z}^{ij\leftrightarrow},c_{z}^{ij\updownarrow},c_{z,z+1}^{ij}>0$ that also depend on the tile 
and define the edge costs accordingly. If an edge belongs to more than one tile, the minimum cost applies. See Figure~\ref{fig:generalcostmodel} for an example.
We allow that two (but not three) consecutive coordinates are identical, i.e., $\xi^i = \xi^{i+1}$ or $\upsilon^j = \upsilon^{j+1}$, in order to model a cheap cost at one x- or y-coordinate only.

With this more general model, one can, for example, punish wires on low layers near the electrical source of a net
(which would lead to poor delays) or implement a discount on reservations as we will describe in detail in Section~\ref{sec:reservations}. 
Moreover, we can set edge costs to infinity outside the area corresponding to the global routing solution so that
the distances in $(G,c)$ reflect necessary detours that are implied by routing in this subgraph. 

Given a finite subgraph $G'=(V',E')$ of $G$ and sets $S,T\subseteq V'$, we look for a shortest (minimum-cost) path from $S$ to $T$
in $G'$ with respect to the cost function $c$.
The graph $G'$ does normally not contain vertices and edges whose use  would result in a conflict to nets routed previously
(an exception will be described at the end of Section~\ref{subsection:implementation}), and it can have
a very complicated structure.
For a goal-oriented path search, we define a potential $\pi(v)$ for every vertex $v\in V'$ by
the distance to $T$ in $G$ (instead of $G'$):
$$\pi(v) := \dist_{(G,c)}(v,T).$$
The idea is that distances in $G$ are much easier to compute than in the subgraph $G'$ (we will see how fast),
but often still give a good lower bound. The reason is that $(G,c)$ has a simple structure, given by the tiles.
 
This allows us to use Dijkstra's algorithm with the reduced costs $c_{\pi}$ in the digraph resulting from $G'$ by orienting every edge in both ways,
since the reduced costs are nonnegative.
Indeed, we have $c_{\pi}((v,w))=c(e)-\pi(v)+\pi(w)=c(e)-\dist_{(G,c)}(v,T)+\dist_{(G,c)}(w,T)\ge 0$ for all $e=\{v,w\}\in E$.
After introducing a super-source $\bar s$ and arcs $(\bar s,s)$ of cost 0 for all $s\in S$, with $\pi(\bar s):=\min\{\pi(s):s\in S\}$,   
every path $P$ from $\bar s$ to $T$ satisfies $c_{\pi}(P)=c(P)-\pi(\bar s) $, so shortest $\bar s$-$T$-paths with respect to $c_{\pi}$
are shortest $\bar s$-$T$-paths (and correspond to shortest $S$-$T$-paths) with respect to $c$.

The better the lower bound $\pi$ on the distance to $T$ in $(G',c)$ is, the fewer vertices will be labeled by Dijkstra's algorithm with reduced costs  $c_{\pi}$;
more precisely all vertices $v$ with $\dist_{(G',c)}(S,v)+\pi(v)<\dist_{(G',c)}(S,T)$ will be processed.

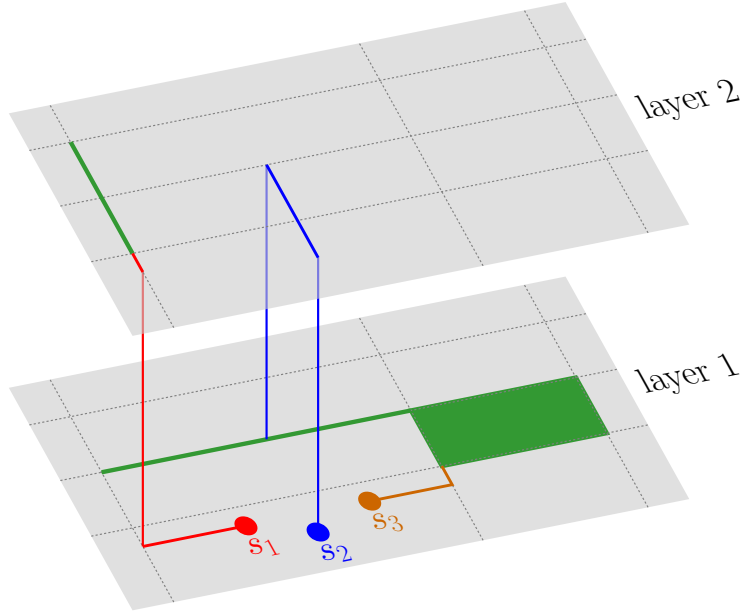
\begin{figure}[b]
 \centering
   \resizebox{0.6\textwidth}{!}
  {
  \begin{tikzpicture}
   \newcommand{\yslant}{0.2} %0.5
   \newcommand{\xslant}{-0.5} %-1
   \newcommand{\yshift}{379}
   \newcommand{\opa}{0.6}
  \begin{scope}[
 yshift=-\yshift,every node/.append style={
 yslant=\yslant,xslant=\xslant},yslant=\yslant,xslant=\xslant
  ]
   \draw (35, 11) node {\fontsize{50}{60}\selectfont layer $1$};
   \fill[black!20,fill opacity=\opa] ( 5,  5) rectangle (32, 17);  
   %target areas
   \fill [darkgreen!80] (22,  9) rectangle (30, 12);
   \draw [darkgreen!80, line width=6pt] (22,  9) rectangle (30, 12);
   
   %Hanan grid
   \foreach \x in { 7, 22, 30 }
   {
    \draw[ultra thick, gray, dashed] (\x,5) -- (\x,17);
   }
   \foreach \y in { 9, 12, 15 }
   {
    \draw [ultra thick, gray, dashed] (5,\y) -- (32,\y);
   }
   
   %query location
   \fill [ red] (12.0,  8.0) circle (0.5cm);
   \draw [ red] (12.2,  6.8) node {\fontsize{50}{60}\selectfont $s_1$};
   \fill [blue] (15.0,  7.0) circle (0.5cm);
   \draw [blue] (15.2,  5.8) node {\fontsize{50}{60}\selectfont $s_2$};
   \fill [orange] (18.0,  8.0) circle (0.5cm);
   \draw [orange] (18.2,  6.8) node {\fontsize{50}{60}\selectfont $s_3$};
   
   %part of shortest path
   \draw [ red, line width = 4pt] ( 7,  8) -- (12,  8);
   \draw [orange, line width = 4pt] (18,  8) -- (22, 8) -- (22,  9);

   %target 
   \draw [darkgreen!80, line width = 6pt] ( 7, 12) -- (22, 12);
  \end{scope}
   
   %    Vias
  \begin{scope}[
 yshift=-\yshift,every node/.append style={
 yslant=\yslant,xslant=\xslant},yslant=\yslant,xslant=\xslant
  ]
   \draw [ red, line width=3pt] ( 7,  8) -- ($( 7,  8) + (6.675,13.35)$);
   \draw [blue, line width=3pt] (15,  7) -- ($(15,  7) + (6.675,13.35)$);
   \draw [blue, line width=3pt] (15, 12) -- ($(15, 12) + (6.675,13.35)$);
  \end{scope}
   
  \begin{scope}[
 yshift=0,every node/.append style={
 yslant=\yslant,xslant=\xslant},yslant=\yslant,xslant=\xslant
  ]
   \draw (35, 11) node {\fontsize{50}{60}\selectfont layer $2$};
   \fill[black!20,fill opacity=\opa] ( 5,  5) rectangle (32, 17);  
   %target areas
   %\fill [darkgreen!80] (22, 9) rectangle (30, 12);
   
   %Hanan grid
   \foreach \x in { 7, 22, 30 }
   {
    \draw[ultra thick, gray, dashed] (\x,5) -- (\x,17);
   }
   \foreach \y in { 9, 12, 15 }
   {
    \draw [ultra thick, gray, dashed] (5,\y) -- (32,\y);
   }
   
   %target 
   \draw [darkgreen!80, line width = 6pt] (7,  9) -- (7, 15);
   
   %part of shortest path
   \draw [blue, line width = 4pt] (15,  7) -- (15, 12);
   \draw [ red, line width = 4pt] ( 7,  8) -- ( 7,  9);
  \end{scope}
  \end{tikzpicture}
  }
  \caption 
  {
    \label{fig:refinedgrid}
    A target set \textcolor{darkgreen}{$T$} consisting of \textcolor{darkgreen}{three target rectangles} on two layers.
    The coarsest partition into tiles that is consistent with $T$ is shown by the dotted lines.
    The figure also shows three possible locations
    (\textcolor{red}{$s_1$}, \textcolor{blue}{$s_2$}, and \textcolor{orange}{$s_3$})
    for which we might be interested in the distance to the closest target.
    The shortest \textcolor{red}{$s_1$-$T$-path}, the shortest \textcolor{blue}{$s_2$-$T$-path}
    and the shortest \textcolor{orange}{$s_3$-$T$-path}
    for a cost function depending only on direction and layer are shown.
    In this example,
    the preference directions of layer $1$ and $2$ are $\leftrightarrow$ and $\updownarrow$, respectively.
    Nevertheless it is cheapest for the \textcolor{orange}{$s_3$-$T$-path}
    to stay on layer $1$ since its vertical segment is very short.
  }
\end{figure}

The target set $T$ can consist of a single vertex (corresponding to a pin), but a pin sometimes covers more than one vertex,
and when constructing Steiner trees from paths we often want to connect to a connected component that contains wires and more than one pin.
We assume that $T$ is represented as the union of~$t$ rectangles,
where a rectangle is a vertex set of the form $\{(x,y,z)\in V \mid \xi^-\le x\le\xi^+,\, \upsilon^-\le y\le\upsilon^+,\, z=\zeta\}$
for some $\xi^-,\xi^+,\upsilon^-,\upsilon^+\in\mathbb{Z}$ and $\zeta\in\{1,\ldots,l\}$.
Often $t$ is small in practice.
While we can deal with complicated target sets, some of our algorithms work best for simple targets (i.e., small $t$).

Sometimes it will be useful to assume that this representation is \emph{consistent} with the partition of $V$ into tiles in the following sense:
each of the $t$ rectangles representing $T$ fits into the grid, i.e., is of the form $\{(x,y,z)\in V \mid \xi^{i^-}\le x\le\xi^{i^+},\, \upsilon^{j^-}\le y\le\upsilon^{j^+},\, z=\zeta\}$
for some indices $i^-$, $i^+$, $j^-$ and $j^+$.
This can be achieved by adding at most $2t$ new x-coordinates~$\xi^i$ and at most $2t$ new y-coordinates~$\upsilon^i$.
We call this procedure \emph{refining the grid with respect to the targets}.
See Figure~\ref{fig:refinedgrid} for an example of the empty grid (i.e., $p=q=0$) refined with respect to several target rectangles.

\subsection{Previous work and our results}

In the simple model without regions (i.e., for $p=q=0$), one can query $\pi$ (i.e., evaluate $\pi(v)$ for a given query location $v\in V$) easily 
in $O(t l^2)$ time without preprocessing; see Proposition~\ref{prop:futurecost_l2_time}.
We show that this can be reduced to $O(t l)$; see Theorem~\ref{thm:futurecost_l_time}.
With a preprocessing time polynomial in $t$ and $l$, we obtain a query time of $O(\log (t + l))$; see Theorem~\ref{thm:futurecost_log_l_time}.
These results will be presented in Sections~\ref{section:NoPreprocessing} and \ref{section:LogarithmicQuery}.

For the more general model, which is the subject of Section~\ref{sec:general_model}, 
Peyer et al.\ \cite{PeyerRautenbachVygen2009} refined the grid with respect to the targets
and showed that then the restriction of $\pi:V\to\mathbb{R}_{\ge 0}$ to $V^{ij}_z$
is the minimum of $k^2$ affine functions for any $i,j,z$, where $k$ is
the number of different horizontal and vertical edge costs, i.e.,
%$k \coloneqq |C^\leftrightarrow| \cdot |C^\updownarrow|$ with
\begin{equation}
\label{eq:def_k}
k \coloneqq \left| \left\{c_{z}^{ijd}\mid
i \in \{0, \dots, p\}, j \in \{0, \dots, q\}, d \in \{\leftrightarrow, \updownarrow\}, z \in \{1, \dots, l\}\right\} \right|. 
\end{equation}
%(In fact instead of $k^2$ they proved as upper bound the possible combinations of costs in horizontal and vertical direction.)
They also showed that all these functions can be computed in
$O((p+t)(q+t)lk^4\log (p+q+l+t))$ time,
allowing $O(k^2)$ time queries after this preprocessing (plus $O(\log (p+q+t))$
to find the region containing the given vertex by binary search; 
here and henceforth $p$ and $q$ refer to the original number of rows and columns, before refining the grid).

We make multiple improvements over the approach of Peyer et al.\ \cite{PeyerRautenbachVygen2009}. 
By considering domination between affine functions with different slopes, we reduce the number of affine functions that are needed to describe the minimum. 
By first computing the distances from the edges on the boundaries of the tiles to the targets, we can compute these affine functions faster. 
Finally, we use a regional query data structure to reduce query time. For any $0 < \epsilon \leq 1$, we can obtain an algorithm with preprocessing time $O((p+t)(q+t)\min\{k,(p+q+1)l\}l^{1+\epsilon}\frac{1}{\epsilon}\log(p+q+l+t))$ and query time $O(\log(p + q + t) + \frac{1}{\epsilon} \log(k + l))$.
See Table~\ref{table:overview} for an overview.

\renewcommand*{\arraystretch}{1.2}
\begin{table}
\small
\begin{center}
\setlength{\tabcolsep}{3pt}
\begin{tabular}{llll}
Model & Preprocessing time & Query time & Reference \\
\hline
simple & -- & $O(tl^2)$ & Proposition~\ref{prop:futurecost_l2_time} \\
simple & -- & $O(tl)$ & Theorem~\ref{thm:futurecost_l_time} \\
simple & $O(t^2l^3\log l)$ & $O(\log (t+l))$ & Theorem~\ref{thm:futurecost_log_l_time} \\
general & $O((p+t)(q+t)lk^4\log (p+q+l+t))$ & $O(\log(p+q+t)+k^2)$ & \cite{PeyerRautenbachVygen2009} \\
general & $O((p+t)(q+t)\min\{k,(p+q+1)l\}l^{1+\epsilon}\frac{1}{\epsilon}\log (p+q+l+t))$ &$O(\log(p+q+t) + \frac{1}{\epsilon}\log(k+l))$ & Corollary~\ref{thm:futurecost_stefan} \\
\end{tabular}
\end{center}
\caption{Various methods to compute $\pi(v)$, possibly after preprocessing. 
The running times depend on the number $t$ of target rectangles, the number $l$ of layers, and in the general model
on the numbers $p$ and $q$ of coordinates that define the $(p+1)(q+1)$ regions, and
on the number $k$ of different horizontal and vertical edge costs (cf.\ \eqref{eq:def_k}). 
Note that $k\le 2(p+1)(q+1)l$.
\label{table:overview}}
\end{table}
\renewcommand*{\arraystretch}{1}
  
Our second contribution is a new approach to incremental routing. 
Rather than trying to correct a given infeasible input routing with local transformations only, 
we compute a new routing from scratch, at least for all nets for which
the input routing does not obey all rules. 
However, in an incremental routing setting most wires will be legal, i.e., do not have a conflict with any other wire.
In order to compute a solution similar to the input where reasonable,
we reserve the space occupied by legal input wires for the respective net and allow to use edges corresponding to
input wires at a discount.
By letting each input wire be a separate tile $V^{ij}_z$, we can model the discount in the cost function $c$
and work with reduced costs efficiently. When most of the input routing can be used, we can find a shortest path 
much faster than without a discount. 

This makes this new approach not only useful for incremental routing, but also for bulk routing. Treating the output
of a track assignment as reservations (wherever it is legal) and then pursuing our new incremental routing approach
can combine the advantages of the two classical bulk routing approaches, 
successive shortest paths and track assignment with local corrections.
We explain our new approach in detail in Section~\ref{section:practical}, where we also show experimental results.

\section{Distances without preprocessing in the simple model}
\label{section:NoPreprocessing}

In the simple model, there is always a shortest path with a very simple structure:

\begin{lemma}
\label{lemma:onlyonehorizontalandonevertical}
Let $c:E\to\mathbb{R}_{>0}$ depend only on direction and layer, and let $r,s\in V$. 
Then there is a shortest path $P$ between $r$ and $s$ in $(G,c)$ that consists of at most
one sequence of horizontal edges, at most one sequence of vertical edges, and hence at most three sequences of vias.
\end{lemma}

\begin{proof}
Let $P$ be a shortest path, and let
$P_{[v,w]}$ and $P_{[v',w']}$ be two maximal subpaths of $P$ in the same direction (all-horizontal or all-vertical), 
say from $v$ to $w$ and from $v'$ to $w'$, respectively, and let $P_{[w,v']}$ be the subpath in between.
Suppose, without loss of generality, that $P_{[v,w]}$ and $P_{[v',w']}$ are horizontal paths and that the cost of an edge of $P_{[v,w]}$ is not more expensive 
than the cost of an edge of $P_{[v',w']}$ (note that these paths may be on different layers).
Then translating $P_{[v',w']}$ by adding $w-v'$ to all its vertices, translating $P_{[w,v']}$ by adding $w'-v'$ to all its vertices,
and swapping these two paths in $P$ yields a walk from $r$ to $s$ 
with one maximal horizontal subpath less and at most the same number of maximal vertical subpaths.
Moreover, the cost does not increase. If the walk is not a path, we can shortcut it to a path.
By induction, the assertion follows.
\end{proof}

Hence, in order to compute a shortest path, we can enumerate the layers on which the horizontal sequence and the vertical sequence are, and which of the two comes first:

\begin{proposition}
\label{prop:futurecost_l2_time}
Let $c:E\to\mathbb{R}_{>0}$ depend only on direction and layer.
Then, without preprocessing, one can compute $\dist_{(G,c)}(s,T)$ for any given $s\in V$ and given $T\subseteq V$ consisting of $t$ rectangles in $O(tl^2)$ time.
\end{proposition}

\begin{proof}
Enumerate over all $t$ rectangles that $T$ is composed of. 
For each such rectangle $R$, we can determine the vertex $r\in R$ that is closest to $s$ (geometrically) in constant time.
Then, for each pair of layers $z^\leftrightarrow,z^\updownarrow\in\{1,\ldots,l\}$, we consider two paths that connect $r$ and $s$.
The first one is composed of the path of vias that goes from $r$ to layer $z^\leftrightarrow$, 
followed by the horizontal path that goes to the x-coordinate of $s$,
followed by the path of vias  that goes to layer $z^\updownarrow$, 
followed by the vertical path that goes to the y-coordinate of $s$,
followed by the path of vias  that goes to $s$. Some of these subpaths can be empty.
The second one is constructed analogously, swapping the roles of $r$ and $s$. By Lemma~\ref{lemma:onlyonehorizontalandonevertical}, 
one of these $2l^2$ paths must be optimal.
\end{proof}

We now show how to improve on this, obtaining a linear dependence on the number of layers:

\begin{theorem}
\label{thm:futurecost_l_time}
Let $c:E\to\mathbb{R}_{>0}$ depend only on direction and layer.
Then, without preprocessing, one can compute $\dist_{(G,c)}(s,T)$ for any given $s\in V$ and given $T\subseteq V$ consisting of $t$ rectangles in $O(tl)$ time.
\end{theorem}

\begin{proof}
Again we enumerate over all $t$ rectangles that $T$ is composed of, and 
for each such rectangle $R$, we determine the vertex $r\in R$ that is closest to $s$ (geometrically) in constant time.
First compute the total cost $c_{z_1,z_2}$ of a path of vias between layer $z_1$ and layer $z_2$ 
for all $z_1,z_2\in\{1,\ldots,l\}$ with $\{z_1,z_2\}\cap\{z_r,z_s\}\not=\emptyset$, 
where $z_r$ and $z_s$ denote the layers of $r$ and $s$, respectively.
This can easily be done in $O(l)$ time.

Now we compute the minimum cost of a path from $r$ to $s$ that (when traversed from $r$ to $s$) 
consists of a path of vias, then a horizontal path, then a path of vias, then a vertical path, then a path of vias.
We will then do the same with exchanging the roles of $r$ and $s$, and the smaller of the two is the distance
between $r$ and $s$ by Lemma~\ref{lemma:onlyonehorizontalandonevertical}.

For each layer $z\in\{1,\ldots,l\}$, consider the vertex $v_z$ on layer $z$ 
whose y-coordinate is the one of $r$ and whose x-coordinate is the one of $s$.
We first compute the distance $\bar d_z$ between $r$ and $v_z$ in the subgraph of $G$ 
that contains no horizontal edges on the layers $1,\ldots,z-1$. 
By Lemma~\ref{lemma:onlyonehorizontalandonevertical}, a shortest path from $r$ to $v_z$ consists of a path of vias, 
a single horizontal path, and another path of vias.
Hence $\bar d_z$ is either the sum of $c_{z_r,z}$ (which we have precomputed) 
and $c^\leftrightarrow_z$ times the difference of the x-coordinates of $r$ and $s$, or 
$\bar d_{z+1}+c_{z,z+1}$ (if $z<l$), whichever is smaller. 
This implies that $\bar d_1,\ldots,\bar d_l$ can be computed in reverse order in total time $O(l)$.

Now we compute the distance $d_z$ from $r$ to $v_z$ in $G$ for all $z\in\{1,\ldots,l\}$. 
We have $d_1=\bar d_1$ by definition. 
For $z\in\{2,\ldots,l\}$, we have $d_z=\min\{\bar d_z, d_{z-1}+c_{z-1,z}\}$,
which allows to compute $d_1,\dots,d_l$ in total time~$O(l)$.
See Figure~\ref{fig:linear_time_distance_computation} for an illustration.

Next we compute the cost of a path from $r$ to $s$ that
consists of a shortest path from $r$ to $v_z$, a vertical path on layer $z$, and a path of vias.
It is given by $d_z+c_{z,z_s}$ plus $c^\updownarrow_z$ times the difference of the y-coordinates of $r$ and $s$,
and thus can now be computed in constant time.
Taking the minimum over all layers $z$ yields a shortest path from $r$ to $s$ by Lemma~\ref{lemma:onlyonehorizontalandonevertical}
and thus completes the proof.
\end{proof}

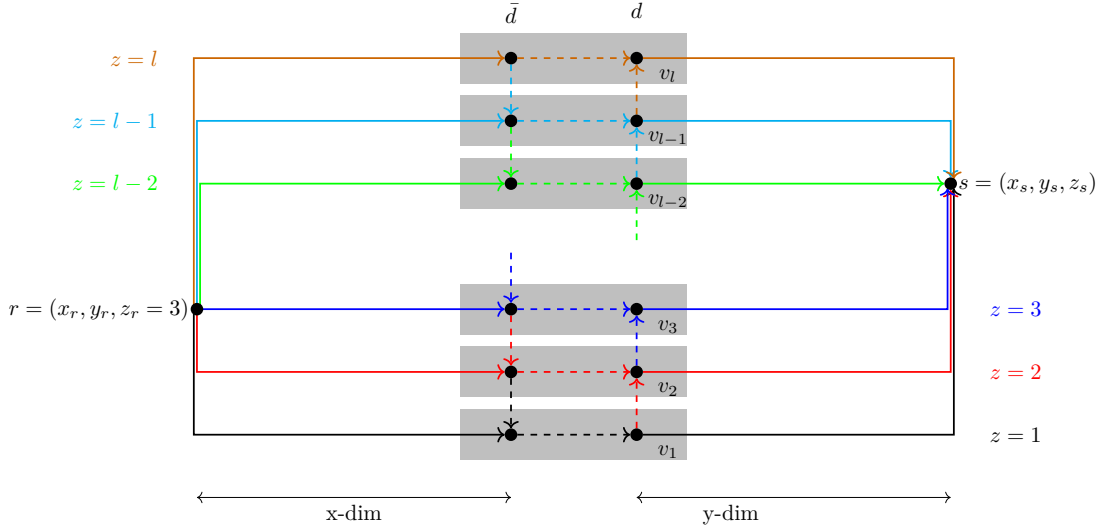
\begin{figure}
 \centering
  \resizebox{0.9\textwidth}{!}
  {
  \begin{tikzpicture}
   \fill [gray!50] (-1.8, 0.6) rectangle (1.8, 1.4);
   \fill [gray!50] (-1.8, 1.6) rectangle (1.8, 2.4);
   \fill [gray!50] (-1.8, 2.6) rectangle (1.8, 3.4);
   \fill [gray!50] (-1.8, 4.6) rectangle (1.8, 5.4);
   \fill [gray!50] (-1.8, 5.6) rectangle (1.8, 6.4);
   \fill [gray!50] (-1.8, 6.6) rectangle (1.8, 7.4);
   %paths
   \draw [->, thick]         (-6.05,3) -- (-6.05, 1) -- (-1.1,1);
   \draw [->, thick, red]    (-6.00,3) -- (-6.00, 2) -- (-1.1,2);
   \draw [->, thick, blue]   (-5.95,3) -- (-5.95, 3) -- (-1.1,3);
   \draw [->, thick, green]  (-5.95,3) -- (-5.95, 5) -- (-1.1,5);
   \draw [->, thick, cyan]   (-6.00,3) -- (-6.00, 6) -- (-1.1,6);
   \draw [->, thick, orange] (-6.05,3) -- (-6.05, 7) -- (-1.1,7);

   \draw [->, thick,        dashed] (-1,1.9) -- (-1,1.1);
   \draw [->, thick, red,   dashed] (-1,2.9) -- (-1,2.1);
   \draw [->, thick, blue,  dashed] (-1,3.9) -- (-1,3.1);
   \draw [->, thick, green, dashed] (-1,5.9) -- (-1,5.1);
   \draw [->, thick, cyan,  dashed] (-1,6.9) -- (-1,6.1);

   %target node
   \fill (-6,3) circle (0.1);
   \draw (-6,3) node[left] {$r = (x_r, y_r, z_r = 3)$};

   %intermediate nodes
   \fill (-1,1) circle (0.1);
   \fill (-1,2) circle (0.1);
   \fill (-1,3) circle (0.1);
   \fill (-1,5) circle (0.1);
   \fill (-1,6) circle (0.1);
   \fill (-1,7) circle (0.1);

   \draw [<->] (-1,0) -- (-3.5,0) node[below] {x-dim} -- (-6,0);

%  \draw [gray, dotted] (-7, 1.5) -- (7, 1.5);
%  \draw [gray, dotted] (-7, 2.5) -- (7, 2.5);
%  \draw [gray, dotted] (-7, 3.5) -- (7, 3.5);
%  \draw [gray, dotted] (-7, 4.5) -- (7, 4.5);
%  \draw [gray, dotted] (-7, 5.5) -- (7, 5.5);
%  \draw [gray, dotted] (-7, 6.5) -- (7, 6.5);

%  \draw [      ] (-6, 8) node[below] {$x=x_r$}; 
%  \draw [      ] ( 6, 8) node[below] {$y=y_s$}; 
   \draw [      ] (-1, 8) node[below] {$\bar{d}$}; 
   \draw [      ] ( 1, 8) node[below] {$d$}; 
   \draw [      ] ( 6.5, 1) node[right] {$z=1$}; 
   \draw [red   ] ( 6.5, 2) node[right] {$z=2$}; 
   \draw [blue  ] ( 6.5, 3) node[right] {$z=3$}; 
   \draw [green ] (-6.5, 5) node[left ] {$z=l-2$}; 
   \draw [cyan  ] (-6.5, 6) node[left ] {$z=l-1$}; 
   \draw [orange] (-6.5, 7) node[left ] {$z=l$};
   %paths
   \draw [<-, thick]         (6.05,4.925) -- (6.05, 1) -- (1.1,1);
   \draw [<-, thick, red]    (6.00,4.9) -- (6.00, 2) -- (1.1,2);
   \draw [<-, thick, blue]   (5.95,4.925) -- (5.95, 3) -- (1.1,3);
   \draw [<-, thick, green]  (5.9,5) -- (1.1,5);
   \draw [<-, thick, cyan]   (6.00,5.1) -- (6.00, 6) -- (1.1,6);
   \draw [<-, thick, orange] (6.05,5.075) -- (6.05, 7) -- (1.1,7);
   \draw [->, thick, red   , dashed] (1,1.1) -- (1,1.9);
   \draw [->, thick, blue  , dashed] (1,2.1) -- (1,2.9);
   \draw [->, thick, green , dashed] (1,4.1) -- (1,4.9);
   \draw [->, thick, cyan  , dashed] (1,5.1) -- (1,5.9);
   \draw [->, thick, orange, dashed] (1,6.1) -- (1,6.9);
   \draw [->, thick,         dashed] (-0.9,1) -- (0.9,1);
   \draw [->, thick, red   , dashed] (-0.9,2) -- (0.9,2);
   \draw [->, thick, blue  , dashed] (-0.9,3) -- (0.9,3);
   \draw [->, thick, green , dashed] (-0.9,5) -- (0.9,5);
   \draw [->, thick, cyan  , dashed] (-0.9,6) -- (0.9,6);
   \draw [->, thick, orange, dashed] (-0.9,7) -- (0.9,7);
   \draw [      ] (1.5, 0.95) node[below] {$v_1$};
   \draw [      ] (1.5, 1.95) node[below] {$v_2$};
   \draw [      ] (1.5, 2.95) node[below] {$v_3$};
   \draw [      ] (1.5, 4.95) node[below] {$v_{l-2}$};
   \draw [      ] (1.5, 5.95) node[below] {$v_{l-1}$};
   \draw [      ] (1.5, 6.95) node[below] {$v_l$};

   %query node
   \fill (6,5) circle (0.1);
   \draw (6,5) node[right] {$s = (x_s, y_s, z_s)$};

   %intermediate nodes
   \fill (1,1) circle (0.1);
   \fill (1,2) circle (0.1);
   \fill (1,3) circle (0.1);
   \fill (1,5) circle (0.1);
   \fill (1,6) circle (0.1);
   \fill (1,7) circle (0.1);
%  \draw (0,1) node[left] {$v_1 = (x_s, y_r, 1)$};
%  \draw (0,2) node[left] {$v_2 = (x_s, y_r, 2)$};
%  \draw (0,3) node[left] {$v_3 = (x_s, y_r, 3)$};
%  \draw (0,5) node[left] {$v_{l-2} = (x_s, y_r, l-2)$};
%  \draw (0,6) node[left] {$v_{l-1} = (x_s, y_r, l-1)$};
%  \draw (0,7) node[left] {$v_l = (x_s, y_r, l)$};

   \draw [<->] (1,0) -- (2.5,0) node[below] {y-dim} -- (6,0);
  \end{tikzpicture}
  }
  \caption 
  { \label{fig:linear_time_distance_computation}
    Illustration of the algorithm described in the proof of Theorem~\ref{thm:futurecost_l_time}.
    The vertices correspond to values computed during the algorithm, the edges correspond to the paths used to derive them. 
    The two vertices in the same gray box are at the same geometric location $v_z = (x_s,y_r,z)$.
    The algorithm starts by computing $\bar d_{l}$ and then derives $\bar d_{z}$ for $z \in \{1,\dots,l-1\}$ as the minimum of $\bar d_{z+1} + c_{z,z+1}$ and the cost of the path using layer $z$ for the horizontal edges.
   Then it propagates the values from bottom to top to compute $d_1, \dots, d_l$.
   For each $z \in \{1,\dots,l\}$ we combine the resulting $r$-$v_z$-path with the paths indicated on the right to get an $r$-$s$-path.
   Finally, we take the cheapest of the $l$ paths to obtain a shortest $r$-$s$-path.
  }
\end{figure}

\section{Logarithmic query time in the simple model}
\label{section:LogarithmicQuery}

%In~\cite{Ahrens2020}, the first author showed 
We will now show
how to achieve $O(\log (t + l))$ query time with polynomial preprocessing time.
We will need two ingredients before we can prove this.
The first ingredient is an efficient algorithm for computing the intersection of three-dimensional half-spaces:

\begin{theorem}[\cite{PreparataMueller1979}] \label{thm:half_space_intersection}
 The intersection of a set of $n$ half-spaces in three-dimensional space can be computed in $O(n \log n)$ time.
 If the intersection is nonempty, it is a convex polyhedron that is presented as a minimal set of inequalities, 
 the cycle of edges surrounding each face, and the coordinates of their endpoints.
\end{theorem}

The second ingredient is an algorithm for solving the planar point location problem, a well-studied problem in computational geometry.
In the following, we regard any connected, closed part of a line in $\Rbb^2$ as a \emph{line segment} and any region whose boundary consists of a finite number of line segments as a \emph{polygon}.

\begin{theorem}[\cite{LiptonTarjan1977}] \label{thm:point_location}
  Let $L$ be a finite set of line segments that intersect only at their endpoints and let $(P_i)_{i \in I}$ be the (open polygonal) connected components of $\Rbb^2 \setminus \bigcup L$.
  Then there is a data structure that requires $O(\lvert L \rvert\log\lvert L \rvert)$ preprocessing time and, given any query point $p \in \Rbb^2$, can then determine an index~$i \in I$ such that $p$ lies in the closure of $P_i$ in $O(\log\lvert L \rvert)$ query time.
\end{theorem}

Point location algorithms that attain the same theoretic guarantees as \cite{LiptonTarjan1977}, but successively improve practical performance and ease of implementation have been described in \cite{Kirkpatrick1983,EdelsbrunnerGuibasStolfi1986,SarnakTarjan1986}.

Theorems~\ref{thm:half_space_intersection} and~\ref{thm:point_location} can be combined in order to obtain a data structure for storing affine functions efficiently such that their pointwise minimum can be queried in logarithmic time:

\begin{lemma} \label{lemma:query_data_structure}
  Let $F$ be a set of affine functions $f \colon \Rbb^2 \to \Rbb$ and $R := [x^-, x^+] \times [y^-, y^+]$ a closed rectangle. Then there is a data structure that requires $O(\lvert F \rvert \log\lvert F \rvert)$ preprocessing time and, given any query point $p \in R$, can then determine the value $\min_{f \in F}f(p)$ in $O(\log\lvert F \rvert)$ query time.
\end{lemma}

\begin{proof}
 We intersect the $\lvert F \rvert$ half-spaces 
 $\{(x,y,\varphi) \in \mathbb R^3 \mid \varphi \leq f(x,y)\}$ for $f \in F$ 
 using Theorem~\ref{thm:half_space_intersection}.
 After projecting the result into the plane, we obtain a subdivision of the rectangle~$R$
 into at most $\lvert F \rvert$ convex polygons and a minimizing function $f \in F$ for each polygon.
 
 Second, we initialize a data structure to solve the point location problem within that subdivision. 
 % (1): E = F + K - 1
 % Nach Eulers't Formula, und weil F das äußere nicht enthält, sonst wäre es - 2
 % (2): K \leq \frac{2}{3}E + \frac{4}{3}
 % Weil die 4 Ecken Grad mindestens 2 und alle anderen Knoten Grad mindestens 3 haben
 % 3\cdot(1) + 3\cdot(2): E \leq 3F + 1
 By Euler's formula, since each vertex (with possible exception of the four corners of $R$) in this subdivision is incident to at least three edges, 
 the subdivision contains at most $3 \lvert F \rvert + 1$ line segments.
 Hence, this preprocessing can be implemented to run in $O(\lvert F \rvert \log\lvert F \rvert)$ time
 by Theorems~\ref{thm:half_space_intersection} and~\ref{thm:point_location}.

 When given a query location $p \in R$, we look up a polygon containing the query location in $O(\log \lvert F \rvert)$ time by Theorem~\ref{thm:point_location},
 and evaluate the function attaining the minimum on that polygon in constant time.
 If the query point is on the boundary of multiple polygons,
 it suffices to evaluate the minimizing function of any one of these polygons.
\end{proof}

We can now prove the main result of this section:

\begin{theorem}\label{thm:futurecost_log_l_time}
 Let $c:E\rightarrow\mathbb R_{> 0}$ depend only on direction and layer,
 and let $T \subseteq V$ consist of $t$ rectangles.
 Then there is a data structure that
 requires $O(t^2 l^{3} \log l)$ preprocessing time and,
 for any given $s \in V$,
 can then determine $\dist_{(G,c)}(s,T)$ in $O(\log(t+l))$ query time.
\end{theorem}
\begin{proof}
 We first interpret our instance as an instance of the general model
 by choosing $p=q=0$, and then refine the grid with respect to the targets,
 which yields $O(t^2l)$ tiles $V^{ij}_z$.
 We compute an independent data structure for each tile.
 
 By Lemma~\ref{lemma:onlyonehorizontalandonevertical}, a shortest path from any $s\in V$
 to any $r\in T$ contains at most one sequence of horizontal and at most one sequence of vertical edges.
 For a fixed set of directions and layers of these sequences, 
 for example north on layer $z^\N$ and west on layer $z^\W$,
 there are at most two such paths, depending on whether we first go north and then west or vice versa,
 having possibly different via costs.
 Given a fixed tile $V^{ij}_z$,
 for every $s = (x_s, y_s, z) \in V^{ij}_z$,
 we are interested in the cost of a shortest such path to $T$,
 which is given by an affine function
 $$f^{(i,j,z,\NW,z^\N,z^\W)}(x_s,y_s) = c^{(i,j,z,\NW,z^\N,z^\W)} - y_s c^{\updownarrow}_{z^\N} + x_s c^\leftrightarrow_{z^\W}$$
 for some constant $c^{(i,j,z,\NW,z^\N,z^\W)}$.
 To obtain the constant, one takes the minimum over all $r = (x_r, y_r, z_r) \in T$ that lie northwest of $V^{ij}_z$,
 considering the sum of
 $y_r c^{\updownarrow}_{z^\N} - x_r c^\leftrightarrow_{z^\W}$ and
 the via cost in the cheaper of the two cases regarding the order of the two directions.
 If there is no $r \in T$ northwest of $V^{ij}_z$, no path of the required structure exists, and the constant can be considered to be $\infty$.

Similarly, an affine function for each of the $1+4l+4l^2$ combinations 
$$(-),(\E,z^\E),(\N,z^\N),(\W,z^\W),(\S,z^\S),(\NE,z^\N,z^\E),(\NW,z^\N,z^\W),(\SW,z^\S,z^\W),(\SE,z^\S,z^\E)$$
can be defined,
and $f^{(i,j,z)}:(x_s,y_s)\mapsto\dist(s,T)$ is the pointwise minimum of these.
Here $(-)$ means that we go to the target only by vias (or we are already at a target).

We next show how to compute all $O(t^2l^3)$ affine functions in total time $O(t^2l^3)$.
We describe this for one combination $(\NW,z^\N,z^\W)$; it works analogously for the other combinations.
Here we need to compute the constants $c^{ij}_z := c^{(i,j,z,\NW,z^\N,z^\W)}$ for all tiles $V^{ij}_z$.
For each layer $z$, we do this from northwest to southeast. We start with
$c^{ij}_z=\infty$ whenever $i=0$ or $j=q$ and then set
$$c^{ij}_z = \min\left\{ c^{i,j+1}_z,\, c^{i-1,j}_z,\, 
\min\left\{ \upsilon^{j+1}c^{\updownarrow}_{z^\N} - \xi^ic^\leftrightarrow_{z^\W} + \min\left\{c_{z,z^\N,z^\W,z'},c_{z,z^\W,z^\N,z'}\right\} \mid (\xi^i,\upsilon^{j+1},z')\in T\right\}\right\}$$
in increasing order of $i-j$, where $c_{z_1,z_2,z_3,z_4}$ is the total cost of the vias to go from layer $z_1$ to layer $z_2$ to layer $z_3$ to layer $z_4$;
see Figure~\ref{fig:affinefunctionstiles} for an illustration.
Since each of the $t$ target rectangles shows up $O(l^3)$ times (once for each combination and each $z$),
all these $O(t^2l^3)$ affine functions can be computed in $O(t^2l^3)$ time.

\begin{figure}
  \centering
  \begin{tikzpicture}
    \draw[fill=black!20] (2,1) rectangle (3,2);

    \draw (0,1) -- (4,1);
    \draw (0,2) -- (4,2);
    \draw (0,3) -- (4,3);
    \draw (1,0) -- (1,4);
    \draw (2,0) -- (2,4);
    \draw (3,0) -- (3,4);

    \node at (0.5,4.5) {\small $i = 0$};
    \node at (1.5,4.5) {\small $i = 1$};
    \node at (2.5,4.5) {\small $i = 2$};
    \node at (3.5,4.5) {\small $i = 3$};
    \node at (-0.5,0.5) {\small $j = 0$};
    \node at (-0.5,1.5) {\small $j = 1$};
    \node at (-0.5,2.5) {\small $j = 2$};
    \node at (-0.5,3.5) {\small $j = 3$};

    \node at (0.5,0.5) {$\infty$};
    \node at (0.5,1.5) {$\infty$};
    \node at (0.5,2.5) {$\infty$};
    \node at (0.5,3.5) {$\infty$};
    \node at (0.5,3.5) {$\infty$};
    \node at (1.5,3.5) {$\infty$};
    \node at (2.5,3.5) {$\infty$};
    \node at (3.5,3.5) {$\infty$};

    \draw[->,thick] (1.5,1.5) -- (2.4,1.5);
    \draw[->,thick] (2.5,2.5) -- (2.5,1.6);
    \fill[darkgreen] (2,2) circle (0.1);
  \end{tikzpicture}
  \hspace{3cm}
  \begin{tikzpicture}
    \draw[fill=black!20] (2,1) rectangle (3,2);

    \draw (0,1) -- (4,1);
    \draw (0,2) -- (4,2);
    \draw (0,3) -- (4,3);
    \draw (1,0) -- (1,4);
    \draw (2,0) -- (2,4);
    \draw (3,0) -- (3,4);

    \node at (0.5,4.5) {\small $i = 0$};
    \node at (1.5,4.5) {\small $i = 1$};
    \node at (2.5,4.5) {\small $i = 2$};
    \node at (3.5,4.5) {\small $i = 3$};
    \node at (-0.5,0.5) {\small $j = 0$};
    \node at (-0.5,1.5) {\small $j = 1$};
    \node at (-0.5,2.5) {\small $j = 2$};
    \node at (-0.5,3.5) {\small $j = 3$};

    \node at (0.5,3.5) {$\infty$};
    \node at (1.5,3.5) {$\infty$};
    \node at (2.5,3.5) {$\infty$};
    \node at (3.5,3.5) {$\infty$};

    \draw[darkgreen,ultra thick] (2,2) -- (3,2);
    \draw[->,thick] (2.5,2.5) -- (2.5,1.6);
  \end{tikzpicture}
  \caption{Computation of the constants $c^{ij}_z$ for a fixed layer $z$ in the proof of Theorem~\ref{thm:futurecost_log_l_time}. The left picture illustrates the propagation for the directions $\NW$, the right picture for $\N$. The constants corresponding to tiles that have no neighboring tiles in northwest or north direction, respectively, are set to $\infty$. When determining the constant $c^{21}_z$, the constants $c^{22}_z$ and $c^{11}_z$ in the $\NW$ case, and only $c^{22}_z$ in the $\N$ case, are taken into account. Moreover, if the marked point in the $\NW$ case or the marked line segment in the $\N$ case belongs to a target rectangle in any layer $z'$, the cost of a path to this point or line contributes to the computation of $c^{21}_z$ as well.}
  \label{fig:affinefunctionstiles}
\end{figure}
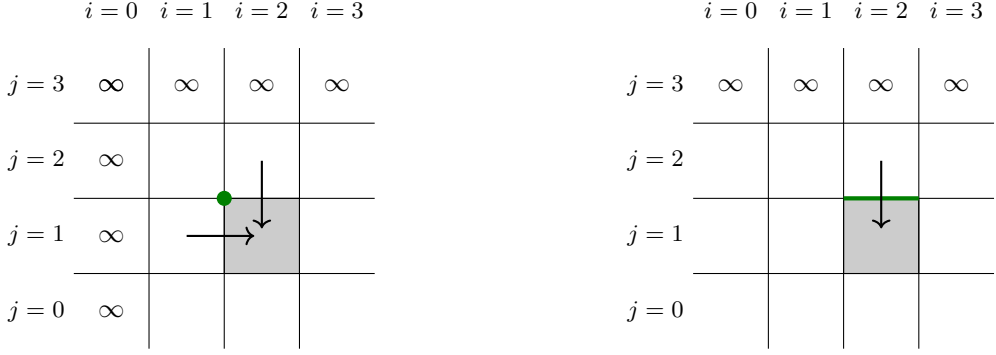

Rather than storing just a list of these $1+4l+4l^2$ affine functions for each $V^{ij}_z$,
we now build a data structure for each tile, using Lemma~\ref{lemma:query_data_structure}, in order to obtain a logarithmic query time. The total preprocessing time required to build these data structures is $O(t^2l^3\log l)$.

 When given a query location $s \in V$, we find the correct tile $V^{ij}_z$
 and hence the correct data structure
 in $O(\log t)$ time by performing two binary searches.
 Then we look up the distance from $s$ to $T$ in $O(\log l)$ time by Lemma~\ref{lemma:query_data_structure}.
 If the query point is on the boundary of multiple tiles,
 it suffices to evaluate the minimizing function of any one of these tiles.
\end{proof}

\section{The general model} \label{sec:general_model}

In this section, we develop an algorithm to compute the potential $\pi(v) = \dist_{(G,c)}(v,T)$ for any $v$ in the general model efficiently after preprocessing. 
We will assume $T$ to be consistent with the grid, i.e., we have already refined the grid if it was not. 
Our preprocessing will work on the horizontal and vertical line segments of the grid, i.e., the sets $\mathrm{Hor}^{ij}_z \coloneqq \left\{(x,y,z) \in V^{ij}_z \mid y = \upsilon^j\right\}$ and $\mathrm{Ver}^{ij}_z \coloneqq \left\{(x,y,z) \in V^{ij}_z \mid x = \xi^i\right\}$. 
The exposition will focus on the horizontal line segments; vertical segments can be handled analogously. 
Our algorithm consists of two preprocessing steps and a query step. The first preprocessing step is a variant of Dijkstra's algorithm. For its correctness, the following observation about the structure of shortest paths, similar to Lemma~\ref{lemma:onlyonehorizontalandonevertical} for the simple model, is essential:

\begin{lemma}
  \label{lemma:pathstructure}
  Let $c : E \to \Rbb_{> 0}$ depend on tile and direction, let $T \subseteq V$ be consistent with the grid, and $s \in V \setminus T$. 
  Then there is a shortest path $P$ from $s$ to $T$ in $(G, c)$ that uses only one type of edges (either horizontal, vertical or via) 
  before entering some tile in which $s$ does not lie.
\end{lemma}

\begin{proof}
By the same argument as in Lemma~\ref{lemma:onlyonehorizontalandonevertical}, there is a shortest path
$P$ such that every subpath of $P$ that is entirely in the interior of one column $\{(x, y, z) \in V \mid \xi^i \le x \le \xi^{i+1}\}$ contains at most one horizontal sequence of edges
and every subpath of $P$ that is entirely in the interior of one row $\{(x, y, z) \in V \mid \upsilon^j \le y \le \upsilon^{j+1}\}$ contains at most one vertical sequence of edges.
It is easy to see that such a path satisfies the claim.
\end{proof}

Our algorithm will first compute $\dist_{(G, c)}(s, T)$ for all $s$ lying in horizontal segments of the grid. 
By applying Lemma~\ref{lemma:pathstructure} inductively, we may assume that the corresponding shortest paths never use a horizontal edge or a via in the interior of a tile.

In the variant of Dijkstra's algorithm that we apply, we do not mark all vertices whose labels are guaranteed to be permanent explicitly, but the set of these vertices at some point in the algorithm is implicitly given: it consists of all vertices that have a smaller label than the one considered in the current iteration. In every iteration, several vertices might be added to this set and several updates of neighboring vertices are performed. This will be useful in order to decrease the required number of iterations significantly compared to the regular Dijkstra's algorithm, while the additional updates in every iteration can be performed with negligible overhead. This is possible because all horizontal edges along some horizontal segment have the same cost, and hence, the corresponding labels can be represented by affine functions that can be stored and updated very efficiently. 
Instead of vertices we will store affine functions in a heap (priority queue), and the \emph{key} of such a function will be the minimum relevant function value.

More precisely, for each horizontal segment $\mathrm{Hor}^{ij}_z$, the algorithm maintains a set $F^{ij}_z$ of affine functions $f \colon [\xi^i, \xi^{i+1}] \to \Rbb_{\ge 0}$ such that each value $f(x)$ corresponds to the cost of a path between $(x, \upsilon^j, z)$ and~$T$. At any point during the algorithm, for every vertex $(x, \upsilon^j, z)$, the minimum value $\min \{f(x) \mid f \in F^{ij}_z\}$ can be considered to be its current label. If $\xi^i = \xi^{i+1}$, we simply store that value. Otherwise we keep only those functions that are not dominated, i.e., attain the pointwise minimum in more than one point. 
The key of a non-dominated function $f \in F^{ij}_z$ is the minimum function value in the interval in which $f$ attains the pointwise minimum.
For storing the sets $F^{ij}_z$, we will use the following result. See Figure~\ref{fig:addfunction} for an illustration.

\begin{lemma}
\label{lemma:functionsets}
  Let $x^-, x^+ \in \Rbb$ be given with $x^- < x^+$.
  We can maintain a data structure that
  \begin{itemize}
    \item stores a set $F$ of affine functions $g : \Rbb \to \Rbb$
      such that the set
      $\{ x \in [x^-, x^+]:\ g(x) = \min_{f \in F} f(x)\}$
      of all points where the function attains the minimum
      is an interval $[x_g^-, x_g^+]$ with $x_g^- < x_g^+$,
    \item stores for every function $g$ the interval $[x_g^-,x_g^+]$
      and the minimum value $\mathrm{key}(g) := \min\{g(x_g^-), g(x_g^+)\}$
      of $g$ on this interval, and
  \item can be updated in $O((D+1)\log (D +|F|))$ time when adding a new function,
    where $D$ is the number of functions that have to be removed from $F$
    because they are dominated.
  \end{itemize}
Moreover, each update increases the key of at most one function that remains in $F$, and never decreases a key.
\end{lemma}

\begin{proof}
  We store the set of affine functions as a binary search tree in which they are sorted by their slopes.
  When a function $g$ is added, we check whether $g$ is dominated and otherwise insert it into the search tree in time $O(\log |F|)$. 
  Then, starting from $g$, we iterate forward and backward, deleting all functions $f$ that are identical to $g$ or dominated by $g$ on the whole interval $[x_f^-, x_f^+]$. Each of these operations requires time $O(\log |F|)$. If a function $f$ is partly dominated by $g$, then its interval and, if necessary, its key are updated. It can be required to shorten two intervals, left and right of the interval of $g$. However, only one key can increase due to the concavity of the pointwise minimum.
\end{proof}

\begin{figure}
  \centering
  \begin{tikzpicture}[thick]
      \fill[yellow!20] (12,4.5) -- (12,1.1) -- (2,4.2) -- (2,4.5) -- (12,4.5);
      \fill[yellow!50] (12,4.5) -- (12,0.4) -- (11,4.5) -- (12,4.5);
      \fill[yellow!50] (12,4.5) -- (12,1.5) -- (9,4.5) -- (12,4.5);
      \fill[yellow!50] (12,4.5) -- (12,2.1) -- (2,4.3) -- (2,4.5) -- (12,4.5);
      \fill[yellow!50] (12,4.5) -- (12,2.6) -- (2,3.6) -- (2,4.5) -- (12,4.5);
      \fill[yellow!50] (3.8,4.5) -- (2,0) -- (2,4.5) -- (3.8,4.5);
  
      \draw[->] (0,0) -- (14,0) node[anchor=west] {$x$};
      \draw[->] (0,0) -- (0,4.5) node[anchor=south] {$f_i(x),g(x)$};
      \draw[gray] (2,4.5) -- (2,-0.2) node[anchor=north] {$x^-$};
      \draw[gray] (12,4.5) -- (12,-0.2) node[anchor=north] {$x^+$};
      \draw (12,0.4) -- (11,4.5) node[right] {$f_5$};
      \draw (12,1.5) -- (9,4.5) node[right=1mm] {$f_4$};
      \draw (12,2.1) -- (2,4.3) node[above right=-1mm] {$f_3$};
      \draw (12,2.6) -- (2,3.6) node[above right=-1mm] {$f_2$};
      \draw (3.8,4.5) -- node[above left=-1mm] {$f_1$} (2,0);
  
      \draw[very thick,blue] (12,1.1) -- node[below] {$g$} (2,4.2);
      
      \draw[dotted] (0,3.27) node[left] {\small new $\text{key}(f_2)$} -- (4.95,3.27) -- (4.95,-0.2) node[below] {\small new $x_{f_2}^+$};
      \draw[dotted] (0,2.97) node[left] {\small old $\text{key}(f_2)$} -- (7.9,2.97) -- (7.9,-0.2) node[below] {\small old $x_{f_2}^+$}; 
      \draw[dotted] (3.4,3.4) -- (3.4,-0.2) node[below] {\small $x_{f_2}^-$}; 
       
  \end{tikzpicture}
  \caption{Update of the data structure described in
    Lemma~\ref{lemma:functionsets} when adding an affine function $g$
    to a set $F=\{f_1,f_2,f_3,f_4,f_5\}$: here $f_3$ and $f_4$ are
    deleted, $x_{f_2}^+$ decreases, $\mathrm{key}(f_2)$ increases, and
    $x_{f_5}^-$ increases.}
  \label{fig:addfunction}
\end{figure}
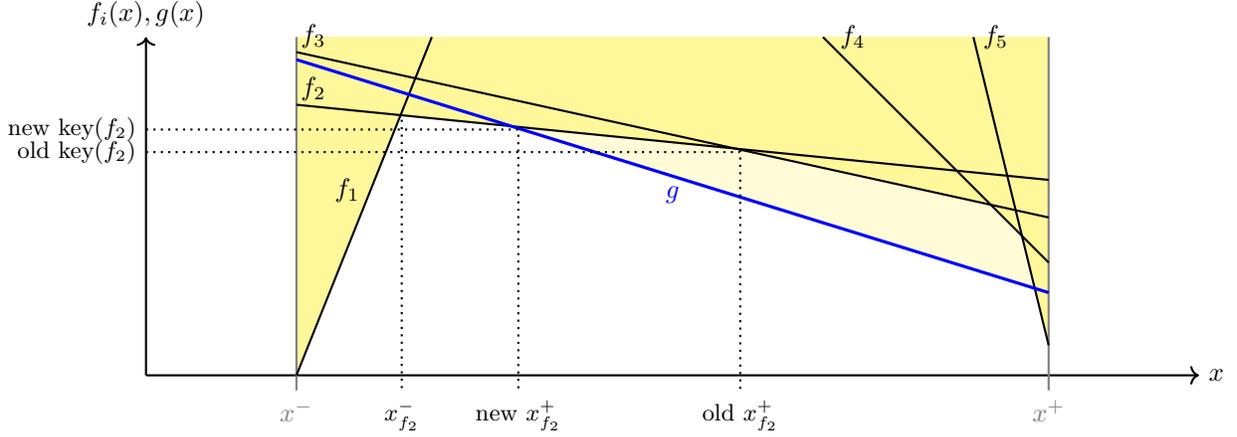

In addition to storing each set $F^{ij}_z$ as specified in Lemma~\ref{lemma:functionsets}, we maintain a binary heap representing all functions in $\bigcup \{F^{ij}_z \mid i \in \{0, \dots, p\}, j \in \{1, \dots, q\}, z \in \{1, \dots, l\}\}$ that have not been processed yet, using the keys defined in Lemma~\ref{lemma:functionsets}. The functions that are added to or removed from some $F^{ij}_z$ must be added to or removed from the heap at the same time, and whenever a key changes, it must be updated also in the heap.

The algorithm starts by initializing $F^{ij}_z \coloneqq \emptyset$ for all $i \in \{0, \dots, p\}$, $j \in \{1, \dots, q\}$, and $z \in \{1, \dots, l\}$. If $\mathrm{Hor}^{ij}_z \subseteq T$, we add the constant function $x \mapsto 0$ to the corresponding set $F^{ij}_z$. If not the whole segment, but one (or both) of its endpoints lies in $T$, we add the affine function describing the distance to this endpoint, i.e., $x \mapsto \min\{c^{ij\leftrightarrow}_z,c^{i(j-1)\leftrightarrow}_z\} \cdot (x-\xi^i)$ or $x \mapsto \min\{c^{ij\leftrightarrow}_z,c^{i(j-1)\leftrightarrow}_z\} \cdot (\xi^{i+1}-x)$.

 In every iteration, a function $f$ with minimum value $\mathrm{key}(f)$ is chosen and removed from the heap. The function $f$ describes the labels of a subset of some horizontal segment $\mathrm{Hor}^{ij}_z$, corresponding to the interval $[x_f^-, x_f^+]$. We now propagate the labels from these vertices to the neighboring horizontal segments by computing at most six new affine functions:
 
 \noindent
 \begin{tabular}{@{}l@{}l@{}}
\text{(down)} & \text{If $z > 1$, add the function $x \mapsto f(x) + \min\{c^{ij}_{z-1,z},c^{i(j-1)}_{z-1,z}\}$ to $F^{ij}_{z-1}$.} \\
\text{(up)} & \text{If $z < l$, add the function $x \mapsto f(x) + \min\{c^{ij}_{z,z+1},c^{i(j-1)}_{z,z+1}\}$ to $F^{ij}_{z+1}$.} \\
\text{(south)} \quad & \text{If $j > 1$, add the function $x \mapsto f(x) + c^{i(j-1)\updownarrow}_z \cdot (\upsilon^j-\upsilon^{j-1})$ to $F^{i(j-1)}_z$.} \\
\text{(north)} & \text{If $j < q$, add the function $x \mapsto f(x) + c^{ij\updownarrow}_z \cdot (\upsilon^{j+1}-\upsilon^j)$ to $F^{i(j+1)}_z$.} \\
\text{(west)} & \text{If $i > 0$, add the function $x \mapsto f(\xi^i) + \min\{c^{(i-1)j\leftrightarrow}_z,c^{(i-1)(j-1)\leftrightarrow}_z\} \cdot (\xi^i-x)$ to $F^{(i-1)j}_z$.} \\
\text{(east)} & \text{If $i < p$, add the function $x \mapsto f(\xi^{i+1}) + \min\{c^{(i+1)j\leftrightarrow}_z,c^{(i+1)(j-1)\leftrightarrow}_z\} \cdot (x-\xi^{i+1})$ to $F^{(i+1)j}_z$.}
\end{tabular}

The algorithm stops when the heap is empty. For an example run of the algorithm, see Figure~\ref{figure:horizontalsegmentalgorithm}. Its correctness, i.e., the fact that, after termination, for all $i \in \{0, \dots, p\}$, $j \in \{1, \dots, q\}$, $z \in \{1, \dots, l\}$, and $(x, \upsilon^j, z) \in \mathrm{Hor}^{ij}_z$, we have $\min\{f(x) \mid f \in F^{ij}_z\} = \dist_{(G, c)}((x, \upsilon^j, z), T)$, follows from the following two lemmas.

\begin{figure}
 \centering
 \begin{tikzpicture}
  \draw ( -0.3, 0.0) -- (14.3, 0.0);
  \draw ( -0.3, 2.0) -- (14.3, 2.0);
  \draw ( 0.0, -0.3) -- ( 0.0, 2.3);
  \draw ( 8.0, -0.3) -- ( 8.0, 2.3);
  \draw (14.0, -0.3) -- (14.0, 2.3);

  \node[anchor=west] at ( 3.2, 0.8) {$c^\leftrightarrow$};
  \node[           ] at ( 4.0, 0.75) {$=$};
  \node[anchor=east] at ( 4.7, 0.8) {$2$};
  \node[anchor=west] at ( 3.2, 1.2) {$c^\updownarrow$};
  \node[           ] at ( 4.0, 1.1) {$=$};
  \node[anchor=east] at ( 4.7, 1.15) {$20$};
  \node[anchor=west] at ( 10.2, 0.8) {$c^\leftrightarrow$};
  \node[           ] at ( 11.0, 0.75) {$=$};
  \node[anchor=east] at ( 11.7, 0.8) {$1$};
  \node[anchor=west] at ( 10.2, 1.2) {$c^\updownarrow$};
  \node[           ] at ( 11.0, 1.1) {$=$};
  \node[anchor=east] at ( 11.7, 1.15) {$10$};

  \fill[color=darkgreen] ( 0.0, 0.0) circle (0.1cm);
  \node[color=darkgreen, anchor=north east ] at ( 0.0, 0.0) {$T$};
  \node[color={rgb:darkgreen,1;darkblue,1;white,1}, anchor=north] at ( 4.0, 0.0) {$x \mapsto 2x$};
  \draw[color={rgb:darkgreen,1;darkblue,1;white,1}, thick, shorten >= 0.01cm, shorten <= 0.01cm] ( 0.0, 0.0) -- ( 8.0, 0.0);
  \node[color=darkblue, anchor=north] at (11.0, 0.0) {$x \mapsto 4 + x$};
  \draw[color=darkblue, thick, shorten >= 0.02cm, shorten <= 0.02cm] ( 8.0, 0.0) -- (14.0, 0.0);
  \node[color={rgb:brown,2;black,2}, anchor=south] at ( 1.5, 2.0) {$x \mapsto 20 + 2x$};
  \draw[color={rgb:brown,2;black,2}, thick, shorten >= 0.02cm, shorten <= 0.02cm] ( 0.0, 2.0) -- ( 3.0, 2.0);
  \node[color=darkred, anchor=south] at ( 5.5, 2.0) {$x \mapsto 26 - 2x$};
  \draw[color=darkred, thick, shorten >= 0.02cm, shorten <= 0.02cm] ( 3.0, 2.0) -- ( 8.0, 2.0);
  \node[color=violet, anchor=south] at (11.0, 2.0) {$x \mapsto 14 + x$};
  \draw[color=violet, thick, shorten >= 0.02cm, shorten <= 0.02cm] ( 8.0, 2.0) -- (14.0, 2.0);

  \tikzset{arrow/.style={-latex, shorten >= 0.02cm, shorten <= 0.02cm}}
  \draw[arrow, gray, bend angle = 45, bend left] ( 0, 0) to ( 1, 0);
  \draw ( 0.5, 0.5) circle (0.2cm) node {$0$};
  \draw[arrow, gray, bend angle = 20, bend right] ( 2, 0) to ( 2, 2);
  \draw ( 1.9, 1.0) circle (0.2cm) node {$1$};
  \draw[arrow, gray, bend angle = 70, bend left] ( 6.7, 0) to ( 8.7, 0);
  \draw ( 7.7, 0.3) circle (0.2cm) node {$1$};
  \draw[arrow, gray, bend angle = 20, bend right] ( 9.4, 0) to ( 9.4, 2);
  \draw ( 9.3, 1.0) circle (0.2cm) node {$2$};
  \draw[arrow, gray, bend angle = 70, bend left] ( 8.7, 2) to ( 6.7, 2);
  \draw ( 7.7, 1.7) circle (0.2cm) node {$3$};
 \end{tikzpicture}
 \begin{caption}
 {
  \label{figure:horizontalsegmentalgorithm}
  Example run of the algorithm
  computing the distance from all horizontal line segments to $T$.
  The instance consists of two horizontally adjacent tiles (i.e., $p=3$, $q=2$, and $l=1$).
  The coordinates of these tiles are $\xi^1=0$, $\xi^2=4$, $\xi^3=7$, $\upsilon^1=0$, and $\upsilon^2=1$.
  The target \textcolor{darkgreen}{$T=\{(0,0,1)\}$} consists of the single point in the bottom left corner of the left tile.
  We disregard the outside tiles (by setting their costs to infinity).
  All other costs are as written in the centers of the respective tiles.
  During the algorithm, five affine functions are added to the four horizontal segments. 
  The horizontal segments are colored by the function attaining the minimum in the end of the algorithm.
  The incoming arrow depicts the propagation by which that function was added
  and is numbered by the iteration of the algorithm (where 0 stands for initialization).
 }
 \end{caption}
\end{figure}

\begin{lemma}
 \label{lemma:pathexists}
 Let $i \in \{0, \dots, p\}$, $j \in \{1, \dots, q\}$, $z \in \{1, \dots, l\}$, and $(x, \upsilon^j, z) \in \mathrm{Hor}^{ij}_z$. If, at any point during the algorithm, an affine function $f$ is added to $F^{ij}_z$, then there is an $(x,\upsilon^j,z)$-$T$-path of cost at most $f(x)$. In particular, $f(x) \ge \dist_{(G, c)}((x, \upsilon^j, z), T)$ holds for all $f \in F^{ij}_z$.
\end{lemma}
\begin{proof}
 We prove the assertion by induction on the order in which the affine functions are added (for all combinations of $i$, $j$, $z$, and $x$ at once). 
 If $f$ is added during the initialization phase, then there is a path with the desired property that only consists of horizontal edges.
Otherwise, $f$ is added during some iteration later on and derived from some $g \in F^{i'j'}_{z'}$. 
Then it suffices to build (possibly zero) edges in one of the six possible directions until a vertex $(x', \upsilon^{j'}, z') \in \mathrm{Hor}^{i'j'}_{z'}$ is reached. 
By the induction hypothesis, there is an $(x', \upsilon^{j'}, z')$-$T$-path of cost at most $g(x')$. 
The propagation ensures that $f(x)$ is at least the cost of this path combined with the straight series of edges.
\end{proof}

\begin{lemma}
 \label{lemma:findspath}
Let $K \in \Rbb_{\ge 0}$ be the key of a function that is chosen in some iteration of the algorithm (or $K = \infty$ if the algorithm has terminated) 
and let $i \in \{0, \dots, p\}$, $j \in \{1, \dots, q\}$, $z \in \{1, \dots, l\}$, and $x \in \{\xi^i, \dots, \xi^{i+1}\}$ such that $\dist_{(G, c)}((x, \upsilon^j, z), T) < K$. 
Then there is a function $f \in F^{ij}_z$ with $f(x) = \dist_{(G, c)}((x, \upsilon^j, z), T)$. 
\end{lemma}
\begin{proof}
Let $s = (x, \upsilon^j, z)$.
We prove the result by induction on $\dist_{(G, c)}(s, T)$.
Note that there can be different $i$ and $j$ with $s \in \mathrm{Hor}^{ij}_z$,
both because $x$ can be a grid coordinate
and because there can be two grid coordinates at $\upsilon^j$.
Our inductive step consists of two parts.
In the first part, we will disregard the given $i$ and $j$
and instead show it is possible to choose $i$ and $j$
such that the inductive hypothesis holds for the given point $s$.

If $s \in T$, then an affine function as desired was added to some $F^{ij}_z$ in the initialization phase.
By Lemma~\ref{lemma:pathexists},
the minimum value of the functions in $F^{ij}_z$ at $x$ cannot decrease.
Because we only remove dominated functions,
there will always be a function in $F^{ij}_z$ attaining this value.

If $s \notin T$, then consider a shortest $s$-$T$-path $P$.
By Lemma~\ref{lemma:pathstructure}, we may assume that
$P$ starts with a straight series of edges to
the first point that lies in a tile not containing $s$.
Denote this point by $s'$.
Consider $i'$, $j'$, and $z'$ with $s' \in \mathrm{Hor}^{i'j'}_{z'}$.
By the induction hypothesis, a function $g$ with $g(x') = \dist_{(G,c)}(s',T)$ is contained in~$F^{i'j'}_{z'}$. 
Using Lemma~\ref{lemma:pathexists}, we derive that $g$ is not strictly dominated in $x'$. 
Hence, $\mathrm{key}(g) \le g(x') = \dist_{(G,c)}(s',T) < \dist_{(G, c)}(s, T) < K$. 
Since keys never decrease during the algorithm by Lemma~\ref{lemma:functionsets},
$g$ was chosen and removed from the heap in a prior iteration.
Now note that it is possible to choose $i$, $j$, $i'$, and $j'$ such that
$s \in \mathrm{Hor}^{ij}_z$ and
the propagation of $g$ adds a function $f$ satisfying $f(x) = c(P)$
to $F^{ij}_z$.
This concludes the first part of the inductive step.

For the second part of the inductive step,
consider all $i$ and $j$ with $s \in \mathrm{Hor}^{ij}_z$.
Note that the possible choices of $i$ and $j$ are independent and consecutive.
We know from the first part that there are some $i_1$ and $j_1$ such that there is
a function $f_1 \in F^{i_1j_1}_z$ with $f_1(x) = \dist_{(G,c)}(s,T) < K$.
Since this function is not dominated in $x$,
we conclude $\mathrm{key}(f_1) < K$.
This means that this function has already been propagated
and has added functions with the same value at $x$
to all adjacent $F^{ij}_z$
for which $s \in \mathrm{Hor}^{ij}_z$.
Successive application concludes the proof.
\end{proof}

Note that both steps are necessary in the above proof.
Figure~\ref{figure:horizontalsegmentalgorithm} contains one such situation:
we are looking for a function $h \in F^{12}_1$
that gives us the correct distance from the top right corner of the left tile to $T$,
i.e., with $h(4) = \dist_{(G,c)}((4,1,1), T)=18$.
Since the vertical cost is lower in the right tile,
the first step gives us the function \textcolor{darkblue}{$g:x \mapsto 4 + x \in F^{21}_1$},
which is propagated to \textcolor{violet}{$f:x \mapsto 14 + x \in F^{22}_1$},
which in turn is propagated to the desired function \textcolor{darkred}{$h:x \mapsto 26 - 2x \in F^{12}_1$}.

To ensure that the algorithm terminates and has the desired running time, we first show:

\begin{lemma}
 \label{lemma:maxfunctions}
 Let $k' \coloneqq \min\{k, (q+1)l\}$.
 The number of slopes of functions ever added to $F^{ij}_z$ is at most
 $2k' + 1$.
 In particular, this also bounds the cardinality of $F^{ij}_z$ at any stage.
\end{lemma}
\begin{proof}
 By induction, the slope of every affine function added to $F^{ij}_{z}$ during the algorithm is either zero or an element of $\{+c^{ij'\leftrightarrow}_{z'}, -c^{ij'\leftrightarrow}_{z'} \mid j' \in \{0, \dots, q\}, z' \in \{1, \dots, l\}\}$. 
 Hence, if there were more than $2k'+1$ functions in $F^{ij}_z$, then there would be two different functions having the same slope. 
 But then one of them is strictly dominated by the other one and would be removed from $F^{ij}_z$, 
 contradicting the specification in Lemma~\ref{lemma:functionsets}.
\end{proof}

In order to achieve the desired running time, 
we bound the number of iterations by~$(p+1)q(2k'+1)l$. Each iteration chooses and removes a function $f$ from the heap. 
Let $(x,\upsilon^j,z) \in \mathrm{Hor}^{ij}_z$ such that $f(x) = \mathrm{key}(f)$. 
By Lemma~\ref{lemma:findspath}, $f(x) = \dist((x,\upsilon^j,z),T)$ and, by Lemma~\ref{lemma:pathexists}, 
each function $g$ added to $F^{ij}_z$ later on satisfies $g(x) \geq \dist_{G,c}((x,\upsilon^j,z), T)$. 
This means no other function with the same slope as $f$
can ever be an element of $F^{ij}_z$ in the future.
Since the number of slopes of functions in $F^{i'j'}_{z'}$ was bounded by $2k'+1$ in Lemma~\ref{lemma:maxfunctions}, we obtain the claimed bound on the number of iterations.

 Finally, we need to implement each iteration in amortized time $O(\log(p+q+l))$. 
 Since each iteration generates at most six new functions and the number of iterations is bounded by $(p+1)q(2k'+1)l$, at most $2(p+1)ql + 6(p+1)q(2k'+1)l$ functions are added in total, including the initialization. 
 This also bounds the size of the binary heap to $O(pq^2l^2)$ such that each heap operation can be performed in time $O(\log(p+q+l))$. 
 By Lemma~\ref{lemma:functionsets}, each added function causes at most one increase-key operation. 
 The total number of deletions is clearly bounded by the total number of insertions. 
 Thus, the total time needed for the heap operations is given by $O(pqk'l\log(p+q+l))$. 
 The same holds for the time required for the updates of the sets $F^{ij}_z$ by Lemma~\ref{lemma:functionsets}. This proves:

\begin{theorem}
 \label{thm:function_labeling}
 There is an algorithm that computes for each horizontal segment $\mathrm{Hor}^{ij}_z$ a set $F^{ij}_z$ of at most $2k'+1$ 
 affine functions such that $\min\left\{f(x) \mid f \in F^{ij}_z\right\} = \dist_{(G,c)}((x,\upsilon^j,z),T)$ 
 for all $(x,\upsilon^j,z) \in \mathrm{Hor}^{ij}_z$. 
 The algorithm can be implemented to run in $O(pqk'l \log (p+q+l))$ time,
 where again $k' = \min\{k, (q+1)l\}$.
 \hfill\qed
\end{theorem}

We will now consider how to implement queries. We do so for all pairs of $i$ and $j$ independently.
If we wanted to execute a query right after executing the algorithm described by Theorem~\ref{thm:function_labeling}
without any further preprocessing,
the best we could do is $O(\log (p + q) + l \log k)$ query time:
we first compute $i$ and $j$ by binary search.
In each of the $4l$ segments
$\mathrm{Hor}^{ij}_{z'}$, $\mathrm{Ver}^{(i+1)j}_{z'}$, $\mathrm{Hor}^{i(j+1)}_{z'}$, and $\mathrm{Ver}^{ij}_{z'}$ 
($z' \in \{1, \dots, l\}$), we then compute the closest point $r$ to the query location $s$.
The distance from $s$ to $r$ can be computed in amortized constant time
and the distance from $r$ to $T$ can be looked up in $O(\log k)$ time
in the data structure storing the affine functions of that segment.

The above can be seen as an evaluation of the minimum
of $O(\min\{k, (p+q+1)l\} l)$ affine functions.
Hence it might be worthwhile building up the data structure 
described by Lemma~\ref{lemma:query_data_structure}
in an additional preprocessing step.
This would speed up our queries to $O(\log (p + q + l))$ time,
however at the cost of an additional $O(\min\{k, (p + q + 1)l\} l \log (p + q + l))$
preprocessing time per tile.

The following result will obtain a trade-off between these two alternatives.
By choosing the trade-off factor $0 < \epsilon \leq 1$ to be a small constant,
we obtain a query time of $O(\log(p + q + l))$
after a preprocessing time which is arbitrarly close to
$O(pq\min\{k,(p+q+1)l\}l\log(p+q+l))$.

\begin{theorem}
\label{thm:general_query_data_structure}
Let $0 < \epsilon \le 1$, let $c : E \to \Rbb_{>0}$ depend on tile and direction, and let $T \subseteq V$ be consistent with the grid. Then there is a data structure that requires $O(pq\min\{k,(p+q+1)l\}l^{1+\epsilon}\frac{1}{\epsilon}\log(p+q+l))$ preprocessing time and, for any given $s \in V$, can then determine $\dist_{(G,c)}(s,T)$ in $O(\log (p + q) + \frac{1}{\epsilon} \log (k + l))$ query time.
\end{theorem}

\begin{proof}
  We first apply Theorem~\ref{thm:function_labeling} to compute $\dist_{(G,c)}(v,T)$ 
  for all $v$ in some $\mathrm{Hor}^{ij}_z$ and (analogously) for all $v$ in some $\mathrm{Ver}^{ij}_z$.
  Now our preprocessing will consider all combinations of $i \in \{0, \dots, p\}$ and $j \in \{0, \dots, q\}$ separately. 
  Given a query location $(x,y,z) \in V^{ij}_{z}$, $i$ and $j$ can be determined in $O(\log(p + q))$ time by binary search. 
  Hence we fix $i$ and $j$ from now on.

  We refer to the union of the up to $4l$ segments
  $\mathrm{Hor}^{ij}_{z'}$, $\mathrm{Ver}^{(i+1)j}_{z'}$, $\mathrm{Hor}^{i(j+1)}_{z'}$, and $\mathrm{Ver}^{ij}_{z'}$ 
  ($z' \in \{1, \dots, l\}$) as the \emph{boundary}. 
  First suppose that a shortest $s$-$T$-path does not touch the boundary. 
  Since $T$ is consistent with the grid, such a path consists of vias only.
  The cost of such vias-only paths can be easily precomputed in $O(pql)$ total preprocessing time,
  allowing for $O(1)$ query time.
  
  Now suppose that a shortest $s$-$T$-path touches the boundary at least once. 
  Consider (without loss of generality) a horizontal segment $H$ of the boundary that is touched first by one such path. By Lemma~\ref{lemma:pathstructure}, we may pick our shortest path such that it starts with a sequence of vias to the layer of $H$, followed by a sequence of vertical edges to $H$. Hence, we can compute the cost of the path until it first touches $H$ as an affine function in $y$, where $s = (x, y, z)$. The rest of the path is a shortest path from some point in $H$ to $T$, so we already computed its cost as a minimum of $O(\min\{k,(q+1)l\})$ affine functions in $x$, by Theorem~\ref{thm:function_labeling}.
 By iterating over each of these affine functions for every horizontal and vertical segment in the boundary, we can express the cost of a shortest $s$-$T$-path as a minimum of $O(\min\{k,(p+q+1)l\}l)$ affine functions in $x$ and~$y$. Note that the only dependence of these functions on $z$ is the cost of the initial via stack. We will exploit this now.
  
  Instead of building a separate data structure for each layer, each involving all these affine functions,
  we distinguish between the cases whether the shortest path from $s$ to $T$ begins 
  without vias, with vias up to a higher level, or with vias down to a lower layer. 
  For the first case, we build up the point location data structure just as in Lemma~\ref{lemma:query_data_structure},
  but each involving only the $O(\min\{k,(p+q+1)l\})$ affine functions on the boundary segments on that layer.
  We call these data structures $D^{= z}$ for $z\in\{1,\ldots,l\}$.
  
  For the other two cases, we build data structures
  $D^{\uparrow [a,b]}$ and $D^{\downarrow [a,b]}$ for some $1\le a\le b \le l$.
  Here $D^{\uparrow [a,b]}$ considers query locations on all layers $z \in \{1, \dots, a\}$ and all boundary segments in the layer range $\{a,\ldots,b\}$, 
  i.e., paths from $(x,y,z)$ that begin with a (possibly empty) via stack from layer $z$ up to some layer $z'\in\{a,\ldots,b\}$ 
  and then proceed via a straight horizontal or vertical path to the boundary.
  Similarly, $D^{\downarrow [a,b]}$ considers query locations on layers $z \in \{b, \dots, l\}$ 
  and boundary segments in the layer range $\{a,\ldots,b\}$. 
  Note that such a data structure involves $O((b+1-a)\min\{k, (p+q+1)l\})$ affine functions, hence,
  by Lemma~\ref{lemma:query_data_structure}, it can be constructed in $O((b+1-a)\min\{k, (p+q+1)l\} \log(p+q+l))$ time
  and then allows for queries in $O(\log(k+l))$ time.

The main advantage is that we can use the same data structure $D^{\uparrow[a,b]}$ for all layers $z \in \{1, \dots, a\}$ because the cost of 
the via stack from layer $z$ to layer $a$ is a constant term that depends only on $z$. 
We can design the data structures so that each affine function on the boundary shows up only in at most $2l^{\epsilon}$ 
(instead of $l$) of these data structures.

To this end, we consider a balanced arborescence $A$ whose leaves are the layers $1,\ldots,l$, such that
$A$ has maximum out-degree $d=\lceil l^{\epsilon}\rceil$ and depth $\lceil\frac{1}{\epsilon}\rceil$, and,
for every vertex $v$ of $A$, the set of leaves reachable from $v$ in $A$ is a consecutive range $L(v)$.
Let $\text{parent}(v)$ denote the parent of $v$ (unless $v$ is the root).
For every vertex $v$ except for the root,
let $L^>(v):= \{z\in L(\text{parent}(v)) \mid z>z' \text{ for all } z'\in L(v)\}$ and
$L^<(v):= \{z\in L(\text{parent}(v)) \mid z<z' \text{ for all } z'\in L(v)\}$.
Then we store $D^{\uparrow L^>(v)}$ and $D^{\downarrow L^<(v)}$, unless this layer range is empty;
see Figure~\ref{figure:partition_layers_new}.

To answer a query for a point on layer $z$, we ask $D^{=z}$ and then
traverse the path from $z$ to the root in the arborescence $A$,
and for each vertex $v$ on that path (except for the root), we ask $D^{\uparrow L^>(v)}$ and $D^{\downarrow L^<(v)}$.
We have to query at most $2\lceil\frac{1}{\epsilon}\rceil +1$ data structures, and each of these queries takes $O(\log(k+l))$ time.

To bound the preprocessing time, we see that each layer $z$ 
appears in the layer range of only $d$ data structures on each level of the arborescence
(one per vertex whose parent's layer range contains $z$), and hence $O(l^{\epsilon}\frac{1}{\epsilon})$ overall.
Hence, the total preprocessing time is $O(\min\{k, (p+q+1)l\} l^{1+\epsilon}\frac{1}{\epsilon} \log(p+q+l))$.
Since we do this for all $i$ and $j$, the theorem follows.

\end{proof}

\begin{figure}
 \centering
 \begin{tikzpicture}
   % This tikz code was generated by partition_layers.py.
   % Don't edit the code, modify script and execute it instead
   \fill(+ 9.00, + 9.00) circle (0.1);
   \draw[gray, anchor=north] (+ 9.00, + 8.90) node {\small $1-39$};
   \fill(+ 9.00, + 9.00) circle (0.1);
   9.0 9.0 3.5999999999999996 7.5
   \draw[->, gray!80] (+ 8.90, + 8.97) -- (+ 3.70, + 7.53);
   \draw(+ 3.60, + 7.50) circle (0.1);
   \draw[ red, anchor=east  ] (+ 3.50, + 7.50) node {\small $\emptyset$};
   \draw[gray, anchor=north] (+ 3.60, + 7.40) node {\small $1-10$};
   \draw[blue, anchor=west ] (+ 3.70, + 7.50) node {\small $11-39$};
   9.0 9.0 7.2 7.5
   \draw[->, gray!80] (+ 8.92, + 8.94) -- (+ 7.28, + 7.56);
   \draw(+ 7.20, + 7.50) circle (0.1);
   \draw[ red, anchor=east  ] (+ 7.10, + 7.50) node {\small $1-10$};
   \draw[gray, anchor=north] (+ 7.20, + 7.40) node {\small $11-20$};
   \draw[blue, anchor=west ] (+ 7.30, + 7.50) node {\small $21-39$};
   9.0 9.0 10.8 7.5
   \draw[->, gray!80] (+ 9.08, + 8.94) -- (+10.72, + 7.56);
   \draw(+10.80, + 7.50) circle (0.1);
   \draw[ red, anchor=east  ] (+10.70, + 7.50) node {\small $1-20$};
   \draw[gray, anchor=north] (+10.80, + 7.40) node {\small $21-30$};
   \draw[blue, anchor=west ] (+10.90, + 7.50) node {\small $31-39$};
   9.0 9.0 14.4 7.5
   \draw[->, gray!80] (+ 9.10, + 8.97) -- (+14.30, + 7.53);
   \draw(+14.40, + 7.50) circle (0.1);
   \draw[ red, anchor=east  ] (+14.30, + 7.50) node {\small $1-30$};
   \draw[gray, anchor=north] (+14.40, + 7.40) node {\small $31-39$};
   \draw[blue, anchor=west ] (+14.50, + 7.50) node {\small $\emptyset$};
   \fill(+10.80, + 7.50) circle (0.1);
   10.8 7.5 5.4 6.0
   \draw[->, gray!80] (+10.70, + 7.47) -- (+ 5.50, + 6.03);
   \draw(+ 5.40, + 6.00) circle (0.1);
   \draw[ red, anchor=east  ] (+ 5.30, + 6.00) node {\small $\emptyset$};
   \draw[gray, anchor=north] (+ 5.40, + 5.90) node {\small $21-23$};
   \draw[blue, anchor=west ] (+ 5.50, + 6.00) node {\small $24-30$};
   10.8 7.5 9.0 6.0
   \draw[->, gray!80] (+10.72, + 7.44) -- (+ 9.08, + 6.06);
   \draw(+ 9.00, + 6.00) circle (0.1);
   \draw[ red, anchor=east  ] (+ 8.90, + 6.00) node {\small $21-23$};
   \draw[gray, anchor=north] (+ 9.00, + 5.90) node {\small $24-26$};
   \draw[blue, anchor=west ] (+ 9.10, + 6.00) node {\small $27-30$};
   10.8 7.5 12.600000000000001 6.0
   \draw[->, gray!80] (+10.88, + 7.44) -- (+12.52, + 6.06);
   \draw(+12.60, + 6.00) circle (0.1);
   \draw[ red, anchor=east  ] (+12.50, + 6.00) node {\small $21-26$};
   \draw[gray, anchor=north] (+12.60, + 5.90) node {\small $27-28$};
   \draw[blue, anchor=west ] (+12.70, + 6.00) node {\small $29-30$};
   10.8 7.5 16.200000000000003 6.0
   \draw[->, gray!80] (+10.90, + 7.47) -- (+16.10, + 6.03);
   \draw(+16.20, + 6.00) circle (0.1);
   \draw[ red, anchor=east  ] (+16.10, + 6.00) node {\small $21-28$};
   \draw[gray, anchor=north] (+16.20, + 5.90) node {\small $29-30$};
   \draw[blue, anchor=west ] (+16.30, + 6.00) node {\small $\emptyset$};
   \fill(+ 9.00, + 6.00) circle (0.1);
   9.0 6.0 5.4 4.5
   \draw[->, gray!80] (+ 8.91, + 5.96) -- (+ 5.49, + 4.54);
   \fill(+ 5.40, + 4.50) circle (0.1);
   \draw[ red, anchor=east  ] (+ 5.30, + 4.50) node {\small $\emptyset$};
   \draw[gray, anchor=north] (+ 5.40, + 4.40) node {\small $24-24$};
   \draw[blue, anchor=west ] (+ 5.50, + 4.50) node {\small $25-26$};
   9.0 6.0 9.0 4.5
   \draw[->, gray!80] (+ 9.00, + 5.90) -- (+ 9.00, + 4.60);
   \fill(+ 9.00, + 4.50) circle (0.1);
   \draw[ red, anchor=east  ] (+ 8.90, + 4.50) node {\small $24-24$};
   \draw[gray, anchor=north] (+ 9.00, + 4.40) node {\small $25-25$};
   \draw[blue, anchor=west ] (+ 9.10, + 4.50) node {\small $26-26$};
   9.0 6.0 12.6 4.5
   \draw[->, gray!80] (+ 9.09, + 5.96) -- (+12.51, + 4.54);
   \fill(+12.60, + 4.50) circle (0.1);
   \draw[ red, anchor=east  ] (+12.50, + 4.50) node {\small $24-25$};
   \draw[gray, anchor=north] (+12.60, + 4.40) node {\small $26-26$};
   \draw[blue, anchor=west ] (+12.70, + 4.50) node {\small $\emptyset$};
   \fill(+16.20, + 6.00) circle (0.1);
   16.2 6.0 14.399999999999999 4.5
   \draw[->, gray!80] (+16.12, + 5.94) -- (+14.48, + 4.56);
   \fill(+14.40, + 4.50) circle (0.1);
   \draw[ red, anchor=east  ] (+14.30, + 4.50) node {\small $\emptyset$};
   \draw[gray, anchor=north] (+14.40, + 4.40) node {\small $29-29$};
   \draw[blue, anchor=west ] (+14.50, + 4.50) node {\small $30-30$};
   16.2 6.0 18.0 4.5
   \draw[->, gray!80] (+16.28, + 5.94) -- (+17.92, + 4.56);
   \fill(+18.00, + 4.50) circle (0.1);
   \draw[ red, anchor=east  ] (+17.90, + 4.50) node {\small $29-29$};
   \draw[gray, anchor=north] (+18.00, + 4.40) node {\small $30-30$};
   \draw[blue, anchor=west ] (+18.10, + 4.50) node {\small $\emptyset$};

 \end{tikzpicture}
 \begin{caption}
 {
  \label{figure:partition_layers_new}
  A possible choice for the aborescence $A$ in the proof of Theorem~\ref{thm:general_query_data_structure} if $l=39$ and $\epsilon = \frac{1}{3}$. For the sake of clarity, the branches below the five white vertices were omitted. Below each node $v$, there is a \textcolor{gray}{gray} label describing the range \textcolor{gray}{$L(v)$}. If $v$ is not the root, there are also a \textcolor{red}{red} label left of $v$ describing the layer range \textcolor{red}{$L^<(v)$} and a \textcolor{blue}{blue} label right of $v$ describing the layer range \textcolor{blue}{$L^>(v)$}.
 } \end{caption}
\end{figure}

\begin{corollary}
\label{thm:futurecost_stefan}
Let $0 < \epsilon \le 1$, let $c : E \to \Rbb_{>0}$ depend on tile and direction, and let $T \subseteq V$, not necessarily consistent with the grid. 
Then there is a data structure that requires $O((p+t)(q+t) \min\{k,(p+q+1)l\}l^{1+\epsilon}\frac{1}{\epsilon}\log(p+q+l+t))$ preprocessing time and, 
for any given $s \in V$, can then determine $\dist_{(G,c)}(s,T)$ in $O(\log (p + q + t) + \frac{1}{\epsilon} \log (k + l))$ query time.
\end{corollary}
\begin{proof}
 Refine the grid with respect to the targets,
 then apply Theorem~\ref{thm:general_query_data_structure}
 to $\epsilon$ and the refined instance.
 Note that this refinement does not increase the number of different costs.
 \end{proof}

 We remark that the same ideas used in this section could be applied in the simple model
 to obtain a faster preprocessing time of $O(t^2l^{2+\epsilon}\frac{1}{\epsilon}\log l)$
 at the cost of a slower query time of $O(\log t + \frac{1}{\epsilon}\log l)$.

\section{Practical aspects} \label{section:practical}

\subsection{Implementation} \label{subsection:implementation}

With some modifications that we will describe below, we implemented the algorithms presented in the previous sections
as part of \BRD \cite{Ahrens2020, AhrensGesterKlewinghausEtc2015,GesterMuellerNiebergPantenSchulteVygen2013,Klewinghaus2022}, 
a detailed router developed at the University of Bonn in joint work with IBM. 
\BRD is the main detailed routing tool used by IBM for the design of its processor chips.

Up to parallelization and conflict resolution, \BRD routes one net after the other.
Each net is routed by iteratively connecting two of its components by a path until the net is fully connected, i.e., one component remains. 
The path search is the algorithmic core of \BRD and requires approximately 80--90\,\% of the total runtime.

To ensure that the layout can be manufactured, certain design rules must be obeyed.
For example, two vias must not be too close to each other even if they belong to the same net.
Shortest paths in the detailed routing graph often correspond to wirings that violate design rules.
Respecting even simple design rules is NP-hard \cite{Ahrens2020}.
\BRD uses a framework consisting of multiple components for avoiding violations.
First, every computed path is handed to a post-processing routine, which attempts to resolve violations locally.
Second, we apply multi-labeling, i.e., we search for shortest paths in a modified graph that can have multiple copies of each vertex (and different edges).
The modifications are done in such a way that certain design rule violations are avoided.
Finally, we impose restrictions to avoid violations at the start and end of a path.
For further implementation details, see \cite{Ahrens2020,AhrensGesterKlewinghausEtc2015}.

All experiments were performed on the same AMD EPYC 7601 machine with 64 CPUs and 1024\,GB main memory using 64 threads. 
Table~\ref{table:testbed} gives an overview of our testbed. 
It consists of nine real-world instances from three recent IBM processor chips in 7\,nm and 5\,nm technology nodes. 
We started all experiments on the same instance from the same snapshot, which was taken right before the detailed routing. 
At this point, a (three-dimensional) global routing and possibly an allowed layer range were already computed for each net.

We use edge costs as they have been developed for many years in real design practice.
They have three main components. The first component is called the base cost.
The base cost does not depend on the net and models the amount of routing resources consumed by a path. Wiring against the preferred direction of a layer (if allowed at all) is ten times as expensive as wiring in the preferred direction. Apart from that, wires in x- and y-direction have the same base cost on all layers. The base cost of vias is chosen such that a via bridge, i.e., a path consisting of two vias on the same layer and a single segment of wiring in preferred direction between them, is cheaper than the direct connection between the two endpoints if and only if it blocks strictly fewer additional tracks. This means that the precise via costs depend heavily on the precise design rules and track patterns of the technology. On the highest layers, the base cost of vias can be more than ten times more expensive than on the lowest layers. This is because the thicker wires on the high layers require a track pattern with larger spacing. 
The second component of the cost function is an additive penalty, increasing the cost of wires outside of the assigned layer range. This is a heuristic approach to avoid timing failures due to wires on lower layers having more resistance. 
Finally, the third component serves to restrict our path search to vertices that are inside the area corresponding to the global routing solution. 
In the context of our general cost model, this can be done by setting the cost to infinity outside this area.

\begin{table}
\centering
{
 \begin{tabular}{llrrrrrrrrrr}
 Chip &  Tech & $l$ &  $|V|$ &            Area & Wires &   Vias &   Nets &   Pins &  Calls &  $t$ & $pql$ \\
      &       &     & $10^9$ & $\mathrm{mm}^2$ &     m & $10^6$ & $10^6$ & $10^6$ & $10^6$ &   $$ &    $$ \\
 \hline
   A1 & 7\,nm &  10 &    1.0 &            0.08 &   2.3 &   3.50 &   0.37 &   1.15 &   0.87 & 1.83 &  3018 \\
   A2 & 5\,nm &  10 &    1.1 &            0.09 &   3.2 &   3.22 &   0.29 &   0.92 &   0.68 & 1.98 &  3180 \\
   A3 & 7\,nm &  10 &    1.4 &            0.10 &   2.8 &   2.75 &   0.26 &   0.78 &   0.56 & 1.94 &  3074 \\
   B1 & 5\,nm &  16 &    4.7 &            0.36 &   9.6 &   6.67 &   0.63 &   1.79 &   1.38 & 2.24 & 10868 \\
   B2 & 7\,nm &  16 &   15.9 &            1.20 &  28.6 &  18.82 &   1.73 &   5.00 &   3.98 & 2.09 & 10410 \\
   B3 & 7\,nm &  16 &   36.7 &            2.77 &  24.8 &  14.04 &   1.37 &   3.73 &   2.75 & 2.51 & 12123 \\
   C1 & 7\,nm &  16 &   94.4 &            6.52 &  26.8 &   1.78 &   0.14 &   0.31 &   0.29 & 6.49 & 28739 \\
   C2 & 7\,nm &  18 &  244.7 &           16.73 &  97.2 &   8.19 &   0.57 &   1.21 &   2.60 & 7.46 & 52052 \\
   D1 & 7\,nm &  16 & 9615.9 &          601.97 & 178.3 &  10.30 &   0.88 &   1.86 &   1.96 & 5.77 & 19065 \\
 \end{tabular}
}
\caption{
 \label{table:testbed}
 Testbed consisting of nine real-world instances.
 Tech refers to the naming of the technology nodes by the foundry.
 The vertex set $V$ of the detailed routing graph is the set of all locations that are on track with respect to at least one wire type on the current layer and would be on track when projected to at least one adjacent layer.
 Both (total length of) Wires and (total number of) Vias refer to the detailed routing computed by BonnRoute.
 Calls is the number of calls to our Dijkstra implementation. Since BonnRoute may try out different side constraints to find a path with as few design rule violations as possible, this does not match the number of computed paths.
 $t$ and $pql$ are the arithmetic mean over all calls to the preprocessing. 
}
\end{table}

\begin{table}
\centering
{
 \begin{tabular}{llr|rr|rr|r}
 &&& \multicolumn{2}{c|}{All Dijkstra calls} & \multicolumn{2}{c|}{Standard Dijkstra calls} & Total \BRD \\
 Chip &         Potential & Preprocessing & Runtime & Labels & Runtime & Labels & Wall time \\
      &                   &          h:mm &    h:mm & $10^9$ &    h:mm & $10^9$ &      h:mm \\
 \hline
 \hline
   A1 &           without &          0:00 &   19:12 &   19.7 &   15:21 &   18.2 &      0:35 \\
   A1 & $\ell_1$-distance &          0:00 &    6:05 &    5.8 &    3:53 &    4.9 &      0:22 \\
   A1 &            simple &          0:00 &    4:28 &    4.0 &    2:32 &    3.2 &      0:21 \\
   A1 &           general &          0:47 &    3:46 &    2.9 &    1:49 &    2.1 &      0:20 \\
 \hline
   A2 &           without &          0:00 &   15:04 &   17.0 &   11:20 &   15.2 &      0:28 \\
   A2 & $\ell_1$-distance &          0:00 &    6:16 &    6.5 &    3:41 &    5.2 &      0:21 \\
   A2 &            simple &          0:00 &    5:37 &    5.6 &    3:06 &    4.4 &      0:20 \\
   A2 &           general &          0:33 &    4:34 &    4.0 &    2:09 &    2.8 &      0:19 \\
 \hline
   A3 &           without &          0:00 &   15:25 &   16.8 &   12:14 &   15.2 &      0:27 \\
   A3 & $\ell_1$-distance &          0:00 &    5:38 &    5.9 &    3:39 &    4.9 &      0:17 \\
   A3 &            simple &          0:00 &    5:02 &    5.3 &    3:10 &    4.3 &      0:17 \\
   A3 &           general &          0:27 &    3:58 &    3.6 &    2:09 &    2.6 &      0:19 \\
 \hline
   B1 &           without &          0:00 &   62:46 &   59.1 &   37:15 &   44.9 &      1:42 \\
   B1 & $\ell_1$-distance &          0:00 &   37:15 &   31.8 &   16:37 &   20.6 &      1:17 \\
   B1 &            simple &          0:03 &   31:43 &   26.7 &   13:14 &   16.9 &      1:13 \\
   B1 &           general &          3:12 &   26:23 &   20.1 &    9:33 &   11.5 &      1:11 \\
 \hline
   B2 &           without &          0:00 &  194:21 &  168.5 &  119:21 &  136.2 &      5:13 \\
   B2 & $\ell_1$-distance &          0:00 &  108:31 &   87.1 &   54:33 &   64.0 &      3:55 \\
   B2 &            simple &          0:10 &   91:56 &   71.7 &   41:35 &   50.2 &      3:40 \\
   B2 &           general &          9:37 &   77:58 &   54.7 &   30:33 &   35.2 &      3:38 \\
 \hline
   B3 &           without &          0:00 &  143:39 &  141.3 &  109:28 &  121.7 &      4:08 \\
   B3 & $\ell_1$-distance &          0:00 &   78:39 &   74.7 &   50:49 &   59.3 &      3:15 \\
   B3 &            simple &          0:07 &   60:15 &   56.7 &   35:33 &   43.6 &      2:57 \\
   B3 &           general &          6:48 &   45:18 &   38.5 &   22:48 &   26.9 &      2:54 \\
 \hline
   C1 &           without &          0:00 &  111:55 &  129.4 &   86:24 &  113.0 &      2:20 \\
   C1 & $\ell_1$-distance &          0:00 &  102:23 &   92.0 &   77:30 &   77.9 &      2:20 \\
   C1 &            simple &          0:00 &   86:21 &   78.9 &   63:16 &   66.1 &      1:56 \\
   C1 &           general &          0:57 &   68:21 &   59.4 &   45:50 &   48.0 &      1:45 \\
 \hline
   C2 &           without &          0:00 & 2331:30 & 1378.3 &  290:01 &  350.9 &     41:40 \\
   C2 & $\ell_1$-distance &          0:00 & 2136:23 & 1110.2 &  309:56 &  274.6 &     38:36 \\
   C2 &            simple &          0:08 & 1997:04 & 1024.7 &  252:46 &  229.7 &     36:36 \\
   C2 &           general &         16:59 & 1942:03 &  935.8 &  206:34 &  180.6 &     36:06 \\
 \hline
   D1 &           without &          0:00 & 5039:08 & 3315.9 & 3937:44 & 2818.2 &     86:14 \\
   D1 & $\ell_1$-distance &          0:00 & 2976:07 & 1452.6 & 2015:02 & 1106.8 &     53:45 \\
   D1 &            simple &          0:06 & 1909:23 &  958.0 & 1102:38 &  663.5 &     36:31 \\
   D1 &           general &          7:16 & 1435:37 &  796.8 &  773:07 &  538.7 &     29:08 \\
 \hline
 \hline
  Sum &           without &          0:00 & 7933:04 & 5246.5 & 4619:12 & 3633.9 &    142:52 \\
  Sum & $\ell_1$-distance &          0:02 & 5457:19 & 2867.0 & 2535:45 & 1618.6 &    104:12 \\
  Sum &            simple &          0:39 & 4191:51 & 2232.0 & 1517:52 & 1082.4 &     83:54 \\
  Sum &           general &         46:40 & 3608:00 & 1916.2 & 1094:36 &  848.9 &     75:44 \\
 \end{tabular}
}
\caption{
 \label{table:future_costs_uniform}
Performance of the following four different feasible potentials on our testbed.
In the rows \textbf{without} potential, each query returns 0 in constant time.
When using \textbf{\boldmath$\ell_1$-distance}, the distances in x- and y-direction are scaled by the minimal $c_z^{\leftrightarrow}$ and $c_z^\updownarrow$, respectively, over all $z$. An $O(l)$ preprocessing computes these two numbers and the total cost $c_{1,z}$ 
of vias from layer 1 to each layer $z$; for the distance of two points on layers $z<z'$, we then use $c_{1,z'}-c_{1,z}$. 
The query returns the $\ell_1$-distance to $T$ in $O(t)$ time by iterating over the targets.
In the \textbf{simple} and \textbf{general} rows, the shortest distance to $T$ in the respective models is returned. 
Here the only difference between the two is that the general model restricts to the area corresponding to the global routing solution 
(outside of it, the costs are infinite). Implementation details for the last two potentials are described in the text.
Runtimes are summed over all 64 threads except for the last column, which shows the total wall time of the
overall \BRD run.
}
\end{table}

Table~\ref{table:future_costs_uniform} compares the performance of path searches with the original edge costs and with the reduced costs using three feasible potentials. 
Each of these three potentials is the distance to $T$ in the same supergraph $G$ of $G'$, but with respect to different edge costs $c$. 
For the simple and the general model, our implementation differs from the description in the previous sections as described below.
\begin{itemize}
 \item Since $t$ is usually small ($3.3$ on average), we iterate over all target rectangles in $T$ for every query. 
 When computing the distance from the query location to one of the target rectangles in the \textbf{simple} model, we know the start and end layer. 
 If we guess the lowest and highest layer used by a shortest path, the distance can be computed in constant time (Proposition~\ref{prop:futurecost_l2_time}). 
 Due to the special structure of our cost function, at most $l$ combinations need to be considered. In a preprocessing step, we compute those combinations that can be optimal for some query locations.
 \item The implementation of the \textbf{general} model uses a modified version of the algorithm used in Theorem~\ref{thm:function_labeling}. Since the distance from a single tile to $T$ can be expressed as a minimum of very few affine functions (on every instance, the average is below $1.2$), it is more efficient to compute these functions instead of the distance from the horizontal and vertical segments to $T$. 
 On the other hand, instead of Lemma~\ref{lemma:functionsets}, we need to use a more complicated data structure to maintain the set of non-dominated functions. We store the convex polygon of points on which each function attains the minimum. This way, insertion can be implemented to run in $O(|F| \log |F|)$ time. During each query, we find the correct tile in $O(\log (p + q + t))$ time using binary search and evaluate all non-dominated functions of that tile on the query location to compute the minimum.
\end{itemize}
The results show that the general potential performs significantly better than the simple potential, which already performs much better than the $\ell_1$-distance potential. Both the number of labels and the runtime improve on every instance, even when considering the additional preprocessing time. The relative improvement differs a lot between different instances. Most of this difference can be explained by some situations in which we get only minor improvements by our potentials:
\begin{itemize}
  \item If no path is found, all reachable vertices in the graph are labeled. None of the potentials show any improvement on these instances. 
  In fact, the path searches without potential are the fastest since they do not need any query time, with a total of $285$ hours (summed over all instances and all 64 threads). 
  With the three potentials, these path searches take a total of $329$, $337$, and $331$ hours, respectively.
 \item After a path search failed, \BRD may perform a backup path search which allows routing through existing wires at high cost (and then would remove (\emph{rip up}) such wires and try to re-route them). 
 Since these rip-up costs are not modeled in any of our potentials, a large portion of the graph may be labeled regardless of which potential is used.
On such instances, the order of the potentials regarding their performance is the same as when looking at all instances, but the relative improvements are much smaller.
\end{itemize}
The column \emph{Standard Dijkstra calls} in Table~\ref{table:future_costs_uniform} excludes these situations and hence shows an even larger gain than the column \emph{All Dijkstra calls}.
The question how to model rip-up costs efficiently when computing potentials remains for future research.

\subsection{Reservations and discounts for incremental routing}
\label{sec:reservations}

In chip design practice, there are two main scenarios where a detailed routing is not computed from
scratch, using just a global routing as input, but in an incremental way, using an approximate detailed routing as input.
The first scenario is when a detailed routing has already been computed, but now a few changes have been made,
for example in order to correct the logical function of the chip or to improve its timing behavior.
The second scenario is when a step in between global and detailed routing is used, typically called
track assignment, that maps the global wires to routing tracks in a way that obeys most --- but not all --- design rules.

In both scenarios of incremental routing, we get an almost feasible detailed routing as part of the input,
and the task is to compute a completely feasible detailed routing by doing only few changes.
While it is not exactly specified what ``few'' means, the motivation is that the input routing has already
been optimized, for example with respect to the timing behavior of the chip; moreover,
one aims at saving runtime. 

The traditional approach to incremental routing is to check for violations of design rules (e.g., wires of different nets overlapping) 
and to try to repair such violations locally, in a relatively small area around that violation.
While this can be parallelized very well, many violations cannot be repaired locally, and then the overall approach may fail
or resort to global path searches as backup. Moreover, if the wiring of a net needs to be repaired in multiple places,
the final result can be quite bad, for example with too many detours to meet timing constraints.

We suggest to repair violations globally but with a preference of using the initial solution.
To this end, we convert any detailed wire in the input to a global wire and possibly a reservation. 
A reservation reserves that space for the particular net. When other nets are routed earlier, this space is blocked.
Therefore, reservations are created only for (parts of) detailed wires that do not conflict with other detailed wires in the input.
For an example situation, see Figure~\ref{fig:reservations}.

\begin{figure}
 \centering
 \resizebox{0.99\textwidth}{!}
 {
  \begin{tikzpicture}
   \newcommand{\yslant}{0.2} %0.5
   \newcommand{\xslant}{-0.5} %-1
   \newcommand{\yshift}{379}
   \newcommand{\opa}{0.6}

   % Layer 1
   \begin{scope}[
    yshift=-\yshift,every node/.append style={
    yslant=\yslant,xslant=\xslant},yslant=\yslant,xslant=\xslant
   ]
    \draw (37, 11) node {\fontsize{50}{60}\selectfont layer $1$};
    \fill[black!20,fill opacity=\opa] (-10,2) rectangle (40,17); 
   
    \draw [ red] (22, 12) rectangle  (23,12.8);
   \draw [ red] (22.5, 13.8) node {\fontsize{50}{60}\selectfont $q$};
   \draw[ultra thick,-{Stealth[scale=2]}] (22.5,11.8) -- (22.5,9.2);
    \fill [ red] (22  , 8.2) rectangle  (23,9);
 %  \draw [ red] (22.5, 6) node {\fontsize{50}{60}\selectfont $q$};
   
       \fill [ darkgreen] (-8.5  , 4.6) rectangle  (-7.5,5.4);
   \draw [ darkgreen] (-8, 3.6) node {\fontsize{50}{60}\selectfont $p_1$};
   
       \fill [ darkgreen] (37.5  , 7.6) rectangle  (38.5,8.4);
   \draw [ darkgreen] (38, 6.6) node {\fontsize{50}{60}\selectfont $p_2$};

       \draw [densely dotted, very thick, darkgreen] (11, 7.6) rectangle  (12,8.4);
   \draw [ darkgreen] (11.5, 6.6) node {\fontsize{50}{60}\selectfont $p_3$};
   \draw[ultra thick,-{Stealth[scale=2]}] (11.5,8.6) -- (11.5,14.4);
       \fill [ darkgreen] (11  , 14.6) rectangle  (12,15.4);
   
    %part of shortest path
    \draw [darkgreen!80, line width = 4pt, dotted ] (12, 8) -- (20, 8);
    \draw [red,          line width = 4pt, dotted     ] (20, 8) -- (25, 8);
    \draw [darkgreen!80, line width = 8pt        ] (25, 8) -- (38, 8);
    \draw [darkgreen!80, line width = 8pt        ] ( -8, 5) -- (13, 5);
    \draw [darkgreen!80, line width = 8pt, dotted] (13, 5) -- (17, 5);
    \draw [blue,         line width = 8pt        ] (12,15) -- (25,15);
   \end{scope}
   
   \begin{scope}[
    yshift=-\yshift,every node/.append style={
    yslant=\yslant,xslant=\xslant},yslant=\yslant,xslant=\xslant
   ]
    \draw [darkgreen!80, line width=4pt, dotted] (17, 5) -- ($(17, 5) + (6.675,13.35)$);
    \draw [darkgreen!80, line width=4pt, dotted] (17, 8) -- ($(17, 8) + (6.675,13.35)$);
    \draw [blue,         line width=8pt] (13, 5) -- ($(13, 5) + (6.675,13.35)$);
    \draw [blue,         line width=8pt] (13,15) -- ($(13,15) + (6.675,13.35)$);
    \draw [blue,         line width=8pt] (25,15) -- ($(25,15) + (6.675,13.35)$);
    \draw [blue,         line width=8pt] (25, 8) -- ($(25, 8) + (6.675,13.35)$);
  \end{scope}
   
  \begin{scope}[
   yshift=0,every node/.append style={
   yslant=\yslant,xslant=\xslant},yslant=\yslant,xslant=\xslant
  ]
   \draw (37, 11) node {\fontsize{50}{60}\selectfont layer $2$};
   \fill[black!20,fill opacity=\opa] (-10,2) rectangle (40,17);

   \draw [darkgreen!80, line width = 4pt, dotted] (17, 5) -- (17, 8);
   \draw [blue, line width = 8pt] (13, 5) -- (13,15);
   \draw [blue, line width = 8pt] (25, 8) -- (25,15);
  \end{scope}
 
   \begin{scope}[
     every node/.append style={scale=2},
     scale=2,
     shift={(-31cm, -3cm)}
   ]
    \fill [white] (19.5,  1.2) rectangle (31.5, 10.5);
    \draw [black] (19.5,  1.2) rectangle (31.5, 10.5);
    \foreach \y in {3.3, 5.1, 6.9, 8.7} { \draw [black] (19.5, \y) -- (31.5, \y); };
    \draw [black] (25.5, 1.2) -- (25.5, 3.3);

    \draw [red, line width = 8pt, dotted] (24, 10.0) -- (27, 10.0);
    \node[red] at (25.5, 9.3){\huge illegal input wire};
    \node[darkgreen!80] at (25.5, 7.5){\huge legal input wire};
    \draw [darkgreen!80, line width = 8pt, dotted] (20,  8.2) -- (23,  8.2);
    \draw [darkgreen!80, line width = 16pt, dotted] (24, 8.2) -- (27, 8.2);
    \draw [darkgreen!80, line width = 16pt] (28, 8.2) -- (31, 8.2);
    \draw [darkgreen!80, line width = 16pt, dotted] (22, 6.4) -- (25, 6.4);
    \draw [darkgreen!80, line width = 16pt] (26, 6.4) -- (29, 6.4);
    \node [darkgreen!80] at (25.5, 5.7){\huge reservation};
    \draw [darkgreen!80, line width = 16pt] (20, 4.6) -- (25, 4.6);
    \draw [blue, line width = 16pt] (26,  4.6) -- (31,  4.6);
    \node at (25.5, 3.9){\huge new solution};
    \fill [darkgreen] (27, 2.2) rectangle  (28, 3.0);
    \fill [red] (29, 2.2) rectangle  (30, 3.0);
    \node  at (28.5, 1.7){\huge new pin position};
    \draw [darkgreen] (21, 2.2) rectangle  (22, 3.0);
    \draw [red] (23, 2.2) rectangle  (24, 3.0);
    \node at (22.5, 1.7){\huge old pin position};
   \end{scope}
  
  \end{tikzpicture}
  }
  \caption 
  {
    \label{fig:reservations}
    Example of re-routing a net using reservations after changes to the input have been made.
    The \textcolor{darkgreen!80}{green} and \textcolor{red}{red} wires connected pins $p_1$, $p_2$, and the old position of $p_3$. 
    Now suppose $p_3$ has been moved and is no longer connected. 
    Moreover, pin $q$ from a different net has been moved and now makes the \textcolor{red}{red} piece of wire illegal.
    Now we want to connect $p_1$, $p_2$, and the new position of $p_3$ and use much of the old wiring. 
    We convert all of the old wiring (green and red) to global wires and add a global wire (not shown) that connects to the new position of $p_3$.
    Next we create reservations.
    Even though all of the \textcolor{darkgreen!80}{green wires} are legal, we may choose to create reservations only for the thick green wires, 
    e.g., if we expect the harm of blocking other nets to outweigh the benefit of keeping them usable for this net.
    When this net is being routed, we may end up discarding the dashed part of the reservations, 
    using the rest of the reservations and adding the \textcolor{blue}{blue wires}. The solid wires then constitute the new routing for this net.
  }
\end{figure}
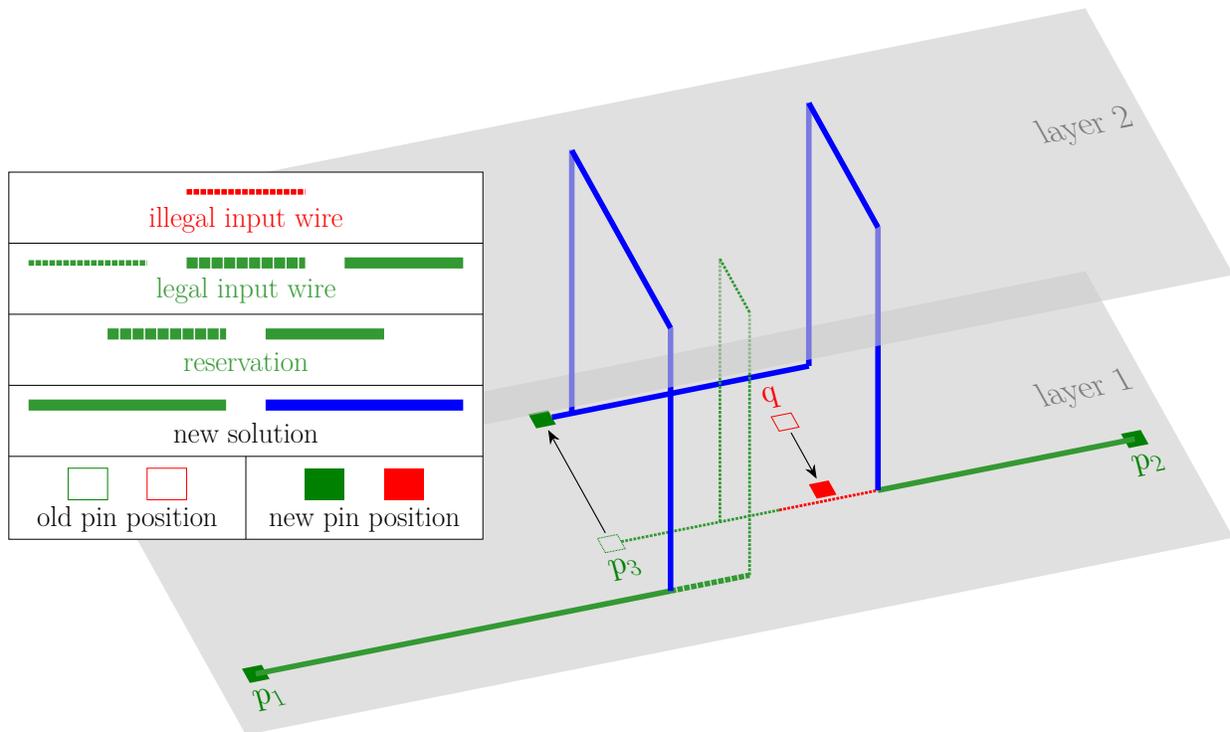

Once a net is routed for which we have created reservations, we would like to encourage, but not force, the net to use
the reserved space. We do this by defining a discount factor $0<\delta<1$ and multiplying all edge costs on a reservation
of that net by $\delta$. 
There are two reasons for using this incentive to route nets similarly as in the input: 
first, during detailed routing, we cannot do a complete timing analysis (this would be too slow), and
the input routing has already been analyzed carefully.
Second, we would hope for a speedup if the reservation serves as a useful guide how to route a net.

However, this speedup does not come automatically. In fact, with the traditional goal-oriented search techniques,
reservations would lead to a slow-down.
For example, if we define the potential $\pi$ to be the $\ell_1$-distance to the nearest target, multiplied by the minimum edge
weight in that direction, then we have to multiply it by $\delta$ if there are any reservations in that direction
(no matter how useful). 
Our generalized framework, however, allows us to refine the grid not only with respect to the targets, but also 
with respect to the reservations, and define individual (discounted) costs on the edges corresponding to reservations.

If we use the output of track assignment as input, this often is a good solution on higher layers, 
where we have mostly longer wires, but much less so on lower layers, which are primarily used for pin access. 
(This is because long wires on low layers have a high resistance and thus poor delay.) 
In this case we may define reservations only on high layers and let the pin access and the short wires be freely determined
by the detailed router.

To evaluate the effect of reservations on incremental routing,
we compare two different algorithms.
Both of them replace the same subset of the input wires by global wires and compute a new solution.
What subset is chosen depends on the scenario and is described below.
The algorithm \emph{no reservations} is our standard bulk routing algorithm,
starting from scratch using these global wires without any additional information.
The algorithm \emph{reservations} creates reservations for all input wires
that are legal, except for short wires that connect only to an illegal input wire.
Both algorithms use the distance in the general model as their potential.
In the algorithm \emph{reservations}, the cost of reservations is multiplied by
a discount factor of $\delta = \frac{3}{4}$ (i.e., $25\%$ discount).
The general model takes this discount into account.

We compare these two algorithms on $15$ instances belonging to two different scenarios.
In the first scenario, we start from an input in which all nets are connected and almost all wires are legal.
More precisely, our input is a snapshot taken right after those bulk routing runs
in Table~\ref{table:future_costs_uniform} that used the general potential.
We do not keep any part of our old solutions fixed,
but replace all the detailed wires in nets we connected by global wires.
The results of this experiment can be seen in Table~\ref{table:reservations-rerun}.

\begin{table}
\centering
{
 \begin{tabular}{llr|rr|r}
 &&& \multicolumn{2}{c|}{Dijkstra} & Total \BRD \\
 Chip &       Algorithm & Preprocessing & Runtime & Labels & Wall time \\
      &                 &          h:mm &    h:mm & $10^9$ &      h:mm \\
 \hline
 \hline
   A1 & no reservations &          0:58 &    3:55 &    2.7 &      0:23 \\
   A1 &    reservations &          1:04 &    2:23 &    0.5 &      0:15 \\
 \hline
   A2 & no reservations &          0:42 &    4:19 &    3.6 &      0:19 \\
   A2 &    reservations &          0:50 &    2:26 &    0.6 &      0:13 \\
 \hline
   A3 & no reservations &          0:40 &    4:12 &    3.5 &      0:18 \\
   A3 &    reservations &          0:45 &    2:06 &    0.6 &      0:12 \\
 \hline
   B1 & no reservations &          4:17 &   23:18 &   17.6 &      1:12 \\
   B1 &    reservations &          4:31 &    8:52 &    3.3 &      0:42 \\
 \hline
   B2 & no reservations &         16:43 &   87:29 &   58.5 &      4:12 \\
   B2 &    reservations &         15:58 &   54:17 &   14.8 &      3:11 \\
 \hline
   B3 & no reservations &          9:19 &   40:29 &   35.3 &      2:51 \\
   B3 &    reservations &         10:32 &   16:46 &    6.3 &      2:24 \\
 \hline
   C1 & no reservations &          1:00 &   70:57 &   59.4 &      2:06 \\
   C1 &    reservations &          1:27 &   25:47 &    9.3 &      1:34 \\
 \hline
   C2 & no reservations &         22:56 & 2114:25 &  902.8 &     39:10 \\
   C2 &    reservations &         21:31 & 1140:54 &  309.6 &     23:35 \\
 \hline
   D1 & no reservations &          8:17 & 1440:19 &  793.2 &     29:09 \\
   D1 &    reservations &         10:54 &  432:58 &  142.0 &     14:03 \\
 \hline
 \hline
  Sum & no reservations &         64:56 & 3789:26 & 1877.1 &     79:44 \\
  Sum &    reservations &         67:34 & 1686:34 &  487.4 &     46:14 \\
 \end{tabular}
}
\caption{
 \label{table:reservations-rerun}
 Performance of different incremental routing algorithms when applied right after bulk routing on all nets. Both runs use the general potential. The cost of reservations is multiplied by $\delta=\frac{3}{4}$.
}
\end{table}

The second scenario in which we evaluate the effect of reservations
is a detailed routing that is no longer legal due to timing optimization.
Table~\ref{table:testbed_eco} gives an overview of this part of the testbed,
consisting of six snapshots taken in a production flow just before incremental detailed routing.
Unlike in the previous scenario, we keep pin-to-pin paths fixed if they consist only of detailed wires.
Any wire not in such a path will be replaced by a global wire.
See Figure~\ref{fig:reservations} for an example.
The performance of both algorithms on these instances is shown in Table~\ref{table:eco_algorithms}.

\begin{table}
\centering
{
 \begin{tabular}{llrrrrrrrrrr}
 Chip &  Tech & $l$ &  $|V|$ &            Area & Wires &   Vias &   Nets &   Pins &  Calls &   $t$ & $pql$ \\
      &       &     & $10^9$ & $\mathrm{mm}^2$ &     m & $10^6$ & $10^6$ & $10^6$ & $10^6$ &    $$ &    $$ \\
 \hline
   b1 & 7\,nm &  16 &    5.7 &            0.40 &   6.6 &   4.22 &   0.41 &   0.48 &   0.04 & 27.95 & 16623 \\
   b2 & 7\,nm &  16 &    7.1 &            0.47 &   9.8 &   6.66 &   0.69 &   0.74 &   0.05 & 29.68 & 20696 \\
   b3 & 7\,nm &  16 &    5.4 &            0.36 &  10.0 &   7.03 &   0.72 &   0.75 &   0.04 & 35.91 & 29935 \\
   b4 & 7\,nm &  16 &    5.0 &            0.36 &  14.8 &  11.19 &   1.04 &   1.06 &   0.02 & 35.34 & 27754 \\
   b5 & 7\,nm &  16 &    6.6 &            0.46 &  15.2 &  13.80 &   1.26 &   1.27 &   0.01 & 30.64 & 21589 \\
   b6 & 7\,nm &  16 &    8.9 &            0.63 &  20.9 &  13.95 &   1.32 &   1.43 &   0.14 & 32.86 & 21524 \\
 \end{tabular}
}
\caption{
 \label{table:testbed_eco}
 Testbed consisting of six snapshots taken during a physical design flow used in production.
 Snapshots were taken after detailed routing based timing optimization,
 right before incremental detailed routing.
 For an explanation of the columns, see Table~\ref{table:testbed}.
 Because the input already contains many valid connections,
 the number of calls to our Dijkstra implementation is often significantly smaller
 than the number of nets.
}
\end{table}

\begin{table}
\centering
{
 \begin{tabular}{llr|rr|r}
 &&& \multicolumn{2}{c|}{Dijkstra} & Total \BRD \\
 Chip &       Algorithm & Preprocessing & Runtime & Labels & Wall time \\
      &                 &          h:mm &    h:mm & $10^9$ &      h:mm \\
 \hline
 \hline
   b1 & no reservations &          0:16 &    4:09 &    3.4 &      0:25 \\
   b1 &    reservations &          0:16 &    2:16 &    1.4 &      0:15 \\
 \hline
   b2 & no reservations &          0:16 &    4:24 &    3.7 &      0:38 \\
   b2 &    reservations &          0:18 &    3:54 &    2.8 &      0:36 \\
 \hline
   b3 & no reservations &          0:23 &    4:03 &    3.0 &      1:23 \\
   b3 &    reservations &          0:24 &    3:38 &    2.3 &      1:21 \\
 \hline
   b4 & no reservations &          0:18 &   11:16 &    4.5 &      0:56 \\
   b4 &    reservations &          0:17 &    6:09 &    2.2 &      0:46 \\
 \hline
   b5 & no reservations &          0:08 &    5:31 &    2.1 &      0:51 \\
   b5 &    reservations &          0:06 &    2:05 &    0.7 &      0:44 \\
 \hline
   b6 & no reservations &          0:49 &   21:26 &   12.1 &      1:15 \\
   b6 &    reservations &          0:50 &   16:05 &    7.8 &      1:09 \\
 \hline
 \hline
  Sum & no reservations &          2:12 &   50:51 &   29.0 &      5:30 \\
  Sum &    reservations &          2:13 &   34:09 &   17.4 &      4:54 \\
 \end{tabular}
}
\caption{
 \label{table:eco_algorithms}
 Performance of different incremental routing algorithms when applied after detailed routing based timing optimization
 on all paths that contain at least one illegal wire. Both algorithms use the general potential.
 The cost of reservations is multipiled by $\delta=\frac{3}{4}$.
 For an explanation of the columns, see Table~\ref{table:future_costs_uniform}.
}
\end{table}

\subsection*{Acknowledgements}

We thank the many other contributors to BonnRoute, in particular Niko Klewinghaus, Christian Roth, and Niklas Schlomberg. 
Thanks also to Lukas K\"uhne, who started the initial implementation of the reservations concept.
We also thank Niklas Schlomberg and the anonymous reviewers for carefully reading a preliminary version of our manuscript. 
Dorothee Henke has partially been supported by Deutsche Forschungsgemeinschaft (DFG) under grant no.~BU~2313/6,
and the other authors under grants EXC~59 and EXC-2047 (Hausdorff Center for Mathematics).
 
\bibliographystyle{acm}
\bibliography{literature.bib}

\begin{thebibliography}{10}

\bibitem{Ahrens2020}
{\sc Ahrens, M.}
\newblock {\em Efficient Algorithms for Routing a Net Subject to {VLSI} Design
  Rules}.
\newblock PhD thesis, University of Bonn, 2020.

\bibitem{AhrensGesterKlewinghausEtc2015}
{\sc Ahrens, M., Gester, M., Klewinghaus, N., M{\"u}ller, D., Peyer, S.,
  Schulte, C., and Téllez, G.}
\newblock Detailed routing algorithms for advanced technology nodes.
\newblock {\em IEEE Transactions on Computer-Aided Design of Integrated
  Circuits and Systems 34}, 4 (2015), 563--576.

\bibitem{AlpertMehtaSapatnekar2008}
{\sc Alpert, C.~J., Mehta, D.~P., and Sapatnekar, S.~S.}
\newblock {\em Handbook of Algorithms for Physical Design Automation}.
\newblock CRC Press, 2008.

\bibitem{batterywala2002track}
{\sc Batterywala, S., Shenoy, N., Nicholls, W., and Zhou, H.}
\newblock Track assignment: A desirable intermediate step between global
  routing and detailed routing.
\newblock In {\em Proceedings of the 2002 IEEE/ACM International Conference on
  Computer-Aided Design\/} (2002), pp.~59--66.

\bibitem{dijkstra}
{\sc Dijkstra, E.~W.}
\newblock A note on two problems in connexion with graphs.
\newblock {\em Numerische Mathematik 1}, 1 (1959), 269--271.

\bibitem{EdelsbrunnerGuibasStolfi1986}
{\sc Edelsbrunner, H., Guibas, L., and Stolfi, J.}
\newblock Optimal point location in a monotone subdivision.
\newblock {\em SIAM Journal on Computing 15}, 2 (1986), 317--340.

\bibitem{GesterMuellerNiebergPantenSchulteVygen2013}
{\sc Gester, M., M{\"u}ller, D., Nieberg, T., Panten, C., Schulte, C., and
  Vygen, J.}
\newblock {BonnRoute}: Algorithms and data structures for fast and good {VLSI}
  routing.
\newblock {\em ACM Transactions on Design Automation of Electronic Systems 18},
  2 (2013), 1--24.

\bibitem{HarNR68}
{\sc Hart, P., Nilsson, N., and Raphael, B.}
\newblock A formal basis for the heuristic determination of minimum cost paths.
\newblock {\em IEEE Transactions of Systems Science and Cybernetics 4\/}
  (1968), 100--107.

\bibitem{Heldetal2018}
{\sc Held, S., Müller, D., Rotter, D., Scheifele, R., Traub, V., and Vygen,
  J.}
\newblock Global routing with timing constraints.
\newblock {\em IEEE Transactions on Computer-Aided Design of Integrated
  Circuits and Systems 37}, 2 (2018), 406--419.

\bibitem{Hetzel1998}
{\sc Hetzel, A.}
\newblock A sequential detailed router for huge grid graphs.
\newblock In {\em Proceedings of Design, Automation and Test in Europe\/}
  (1998), IEEE, pp.~332--338.

\bibitem{Kirkpatrick1983}
{\sc Kirkpatrick, D.}
\newblock Optimal search in planar subdivisions.
\newblock {\em SIAM Journal on Computing 12}, 1 (1983), 28--35.

\bibitem{Klewinghaus2022}
{\sc Klewinghaus, N.}
\newblock {\em Efficient Detailed Routing on Optimized Tracks}.
\newblock PhD thesis, University of Bonn, 2022.

\bibitem{LawLP83}
{\sc Lawler, E., Luby, M., and Parker, B.}
\newblock Finding shortest paths in very large networks.
\newblock In {\em Proceedings of Graph-Theoretic Concepts in Computer
  Science\/} (1983), M.~Nagl and J.~Perl, Eds., Trauner, Linz.

\bibitem{LiptonTarjan1977}
{\sc Lipton, H.~J., and Tarjan, R.~E.}
\newblock Applications of a planar separator theorem.
\newblock {\em 18th Annual IEEE Symposium on Foundations of Computer Science\/}
  (1977), 162--170.

\bibitem{MuellerRadkeVygen2011}
{\sc Müller, D., Radke, K., and Vygen, J.}
\newblock Faster min–max resource sharing in theory and practice.
\newblock {\em Mathematical Programming Computation 3}, 1 (2011), 1--35.

\bibitem{PeyerRautenbachVygen2009}
{\sc Peyer, S., Rautenbach, D., and Vygen, J.}
\newblock A generalization of {D}ijkstra's shortest path algorithm with
  applications to {VLSI} routing.
\newblock {\em Journal of Discrete Algorithms 7}, 4 (2009), 377--390.

\bibitem{PreparataMueller1979}
{\sc Preparata, F.~P., and Müller, D.~E.}
\newblock Finding the intersection of n half-spaces in time {$O(n \log n)$}.
\newblock {\em Theoretical Computer Science 8}, 1 (1979), 45--55.

\bibitem{Rub74}
{\sc Rubin, F.}
\newblock The {L}ee path connection algorithm.
\newblock {\em IEEE Transactions on Computers 23\/} (1974), 907--914.

\bibitem{SarnakTarjan1986}
{\sc Sarnak, N., and Tarjan, R.}
\newblock Planar point location using persistent search trees.
\newblock {\em Communications of the ACM 29}, 7 (1986), 669--679.

\bibitem{sarrafzadeh1994restricted}
{\sc Sarrafzadeh, M., and Lee, D.-T.}
\newblock Restricted track assignment with applications.
\newblock {\em International Journal of Computational Geometry \& Applications
  4}, 1 (1994), 53--68.

\bibitem{Tellez2016}
{\sc Tellez, G., Hu, J., and Wei, Y.}
\newblock Routing.
\newblock In {\em Electronic Design Automation for {IC} Implementation, Circuit
  Design, and Process Technology}, L.~Lavagno, I.~L. Markov, G.~Martin, and
  L.~K. Scheffer, Eds. CRC Press, 2016.

\end{thebibliography}
  
\end{document}